\documentclass[11pt, letterpaper]{article}
\usepackage{cite}
\usepackage{amsmath,amssymb,amsfonts}
\usepackage{algorithmic}
\usepackage{graphicx}
\usepackage{textcomp}
\usepackage{amsthm}
\usepackage{mathptm}   
\usepackage[capitalize]{cleveref}
\usepackage[margin=1in]{geometry}
\usepackage{subfig}
\usepackage{booktabs}
\usepackage[usenames,dvipsnames]{xcolor}
\usepackage{soul}
\usepackage{tikz}
\usepackage{tabu}
\usepackage{array}
\usepackage{booktabs}
\usepackage{multirow}
\usepackage{makecell}
\makeatletter
\newcommand*{\rom}[1]{\expandafter\@slowromancap\romannumeral #1@}
\makeatother

\newtheorem{lemma}{Lemma}
\newtheorem{remark}{Remark}

\usepackage{authblk}

\title{Proportional Power Sharing Control of Distributed Generators \\in Microgrids}
\author[1]{Farzad Aalipour \thanks{farzad.aalipour@knights.ucf.edu}}
\author[2]{Tuhin Das \thanks{tuhin.das@ucf.edu}}
\affil[1,2]{University of Central Florida \\Department of Mechanical and Aerospace Engineering\\ Orlando, Florida, 32816 }
\date{}

\begin{document}

\maketitle
\vspace{-0.5in}
\section*{Abstract}
This research addresses distributed proportional power sharing of inverter-based Distributed Generators (DGs) in microgrids under variations in maximum power capacity of DGs. A microgrid can include renewable energy resources such as wind turbines, solar panels, fuel cells, etc. The intermittent nature of such energy resources causes variations in their maximum power capacities. Since DGs in microgrids can be regarded as Multi-Agent-Systems (MASs), a consensus algorithm is designed to have the DGs generate their output power in proportion to their maximum capacities under capacity fluctuations. A change in power capacity of a DG triggers the consensus algorithm which uses a communication map at the cyber layer to estimate the corresponding change. During the transient time of reaching a consensus, the delivered power may not match the load power demand. To eliminate this mismatch, a control law is augmented that consists of a finite-time consensus algorithm embedded within the overarching power sharing consensus algorithm. The effectiveness of the distributed controller is assessed through simulation of a microgrid consisting of a realistic model of inverter-based DGs. \vspace{-0.5mm}
\section*{Keywords}
Consensus, Eigenvalue Perturbation, Distributed Control, Finite-Time Consensus, Inverter-based Microgrid, Proportional Power Sharing, Renewable Energy, Transient Control.
\vspace{-4mm}

\section{Introduction}
\label{sec:introduction} 
Environmentally sustainable electrical energy production depends on renewable energy resources. In this regard, significant amount of researches have been undertaken within the past few decades \cite{rocabert2012control, he2017simple, du2018distributed}. Conventionally, control of electric power systems and the main power grid was accomplished through a few central controllers. Through emerging renewable energy plants, intelligent loads located in the demand side and computational advances, distributed energy production and management has become viable. DGs as distributed energy production units, together with local loads which are distinct from the main power, is called a microgrid. Microgrids operate in two different operational modes called grid-connected and islanding. A microgrid is said to work in grid-connected mode when it is connected to the main grid via a tie line at the point of common coupling (PCC) where there exists bidirectional power flow from or into the main grid \cite{zhang2012energy}. In contrast, microgrids in islanding mode generates power for local loads \cite{han2017review}. To deploy small-scale DGs including photovoltaic (PV) cells, wind turbines, fuel cells and energy storage systems (ESSs) in microgrids, power electronics inverters are vital interfaces which connect DGs to the power buses \cite{barklund2008energy}.

There have been extensive studies conducted on control of inverter-based microgrids during the past decades \cite{deng2016enhanced,gui2018improved}. The control strategies can be classified in different categories including frequency and voltage control or power control \cite{fan2016distributed}. Applying these methods also depends on the microgrid's mode of operation. For instance, in the grid-connected mode the frequency and voltage are imposed by the main grid. However, voltage and frequency control are vital in islanding mode.

In this work, power sharing control of microgrids in grid-connected mode is studied \cite{mahmud2014robust}. The problem of power sharing has been studied from the aspect of equal power sharing in \cite{guerrero2005output,guerrero2007decentralized}. Since DGs possess different capacities, the DGs with higher capacities can share more power than the DGs with lower capacities. The power sharing problem becomes challenging under intermittent nature of power resources. Intermittency causes fluctuations in maximum capacity of DGs, which leads to changes in their output power. Thereby, the total power fluctuates, and the load power may not always be maintained. These fluctuations can be addressed by deploying electrical energy storage (EES) or managing the DGs to flexibly address the variations in their capacities.

Other approaches proposed to address the power sharing problem in microgrids can be categorized either as proportional power sharing \cite{fan2016distributed,he2014consensus,aalipour2018proportional},
or economic dispatch problem (EDP) \cite{chen2015distributed}. The studies \cite{aalipour2018proportional} and \cite{zhong2011robust} have proposed techniques for proportional power sharing. Here, proportional power sharing is defined as sharing the load among DGs such that each individual DG shares a fraction of the load in proportion to its maximum capacity. A distributed droop control scheme based on nonlinear state feedback proposed in \cite{fan2016distributed} guarantees that DGs share reactive power proportionally. In \cite{schiffer2015voltage},
through a distributed voltage control, the active and reactive power are shared proportionally, for a microgrid with inductive impedance loads. In addition, \cite{cai2018distributed} formulates the proportional power sharing as a tracking problem and solves it for grid-connected spatially concentrated microgrids. However, \cite{fan2016distributed} and \cite{schiffer2015voltage} study islanding mode, and none of \cite{fan2016distributed,schiffer2015voltage,cai2018distributed} have covered the power mismatch during transient time of their proposed strategies. 

On the other hand, EDP is a method to control the power flow among different DGs optimally, where optimality implies minimizing a quadratic performance index assigned to each DG as the cost of their generated power. EDP has been studied through different techniques including the population dynamic method \cite{pantoja2011population}, and the lambda iteration \cite{lin1984hierarchical}. While these methods have been formulated within a centralized control framework, distributed version of EDP can be found in \cite{chen2015distributed,kar2012distributed}.

Motivated by systems with cyber-physical layers, the power sharing control in this study is devised in two layers. The physical layer that consists of DGs, loads, measurement units, etc., is where the power control loop of each DG is established to track the input power command issued from the cyber layer. DGs have their corresponding agents in the cyber layer. Thus, the ideas of MASs can be utilized to establish the DGs' controllers and their interactions. The agents communicate through a communication network in the cyber layer. 

The agents can choose different strategies to control the DGs including centralized, decentralized or distributed formats. When the DGs are located in a small region, it is viable to apply centralized controllers. As the number of DGs increases, while geographically scattered in a wide area, applying the centralized controllers faces deficiencies due to some reasons; Firstly, the centralized controller is not reliable due to the dependency of the DGs on a single controller where its malfunctions deteriorate the performance of the microgrid or may result in instability. Besides, in centralized coordinated control, transferring data to a control center and issuing control signals back to DGs require high bandwidth communication, which is not economically efficient, or technically secure, and is prone to failure \cite{lou2016distributed}. On the contrary, distributed control techniques require considerably lower bandwidth which makes the communications among the DGs economically viable. Decentralized controllers are applicable locally, however it does not exploit cooperation of DGs \cite{vaccaro2011decentralized}. Therefore, they may not perform efficiently where the global information and cooperation is required. In contrast, a distributed control scheme encompasses the plug and play feature, which it makes it more flexible compared to centralized and decentralized controllers  \cite{liu2018game}. The centralized control scheme depends on global information while DGs in distributed control exchange information exclusively with the DGs in their neighborhood. In this study, the well-known consensus algorithm is utilized to design a distributed controller for the power sharing control problem. 

Considering a large number of DGs scattered in a wide area, it takes agents time to transfer all the required signals. Therefore, it is inevitable to have communication delays in the distributed controllers \cite{chaudhuri2004wide}. The ranges of these delays are from tens to hundreds of milliseconds \cite{wu2004evaluation}. The delays may result in prolonging the convergence time of consensus algorithm and potentially lead to microgrid instability  \cite{zhong2011robust}. The delays can be reduced through increasing the convergence rate of consensus algorithms utilizing approaches including multiplying the weights of the communication graph with a large constant, or through an optimization of the weights \cite{xiao2004fast}. 


This study considers a microgrid operating in grid-connected mode using proportional power sharing. Proportional power sharing makes the microgrid adaptable to intermittency of power sources. The grid-connected operation enables the microgrid to transmit excess power to the main grid, while relying on it for frequency and voltage control. The contributions of this paper are:

1) A distributed consensus algorithm is designed, by which the DGs are able to estimate the microgrid’s power capacity under perturbations in the power capacity of individual DGs. Convergence rate of the algorithm is studied and bounds on allowable perturbations are derived based on practical constraints.

2) Multiple proportional power sharing strategies are proposed, to meet the demanded power as consensus is reached. The strategies are executed in a distributed manner. They ensure that the microgrid satisfies load power variations dynamically and allow excess power to be transmitted to the main grid.

3) Items (1) and (2) enable proportional power sharing. However, during convergence of the consensus algorithm, a power mismatch occurs between the generated and demanded power. Although the rate of convergence can be increased, the transient power mismatch remains inevitable. To eliminate this mismatch, a fully distributed finite-time consensus algorithm, based on \cite{charalambous2015distributed}, is additionally augmented.

4) A realistic simulation is conducted, using MATLAB's Simscape toolbox, with a clear primary controller scheme and corresponding parameters. Grid and DG parameters used in simulations are also given. The simulations and the results can be reproduced by readers, allowing further enhancements in future research.

The rest of this paper is organized as follows. The preliminary definitions of technical terms are explained in section two. Then, proportional power sharing is defined in the third section. In the fourth section, the consensus algorithm is developed through which the DGs are able to update their information about the total microgrid power capacity following a change in a DG's capacity. The overarching consensus algorithm and the embedded transient controller are proposed and elaborated in the same section. Fifth section discusses the cyber and physical layers which control the output power of DGs. Next, simulation results are provided in section six to illustrate the effectiveness of the proposed control plan in response to different variations in capacity of a DG. Finally, concluding remarks are provided and references are listed. \vspace{-2mm}

\section{Preliminary Definitions} \label{Sec: Preliminary Definitions} \hspace{-2mm}
We define the graph $\mathcal{G}$ as the set pair $(\mathcal{\nu},\mathcal{\varepsilon})$ having vertices set $\mathcal{\nu}$ and edge set $\mathcal{\varepsilon}$. Let the number of vertices in $\mathcal{G}$ be $N$, and let the set $\mathcal{\varepsilon}$ consist of the vertices pairs $(i,j)$ for which there exists an edge that connects $j$ to $i$, with $i,j=1,2,\cdots,N$ and $i \neq j$. The intended graph in this study is undirected or bidirectional graph, where the signals flow along edges in both directions, i.e. if $(i,j) \in \mathcal{\varepsilon}$, then $(j,i) \in \mathcal{\varepsilon}$. The adjacency matrix associated with the graph is $A=[a_{ij}] \in R^{N \times N}$ where each element $a_{ij}>0$ if $(i,j) \in \mathcal{\varepsilon}$, otherwise $a_{ij}=0$. As stated above, in the bidirectional graph $\mathcal{G}$, if $(i,j) \in \mathcal{\varepsilon}$, then $(j,i) \in \mathcal{\varepsilon}$, and $a_{ij} = a_{ji}$. Then $A$ is symmetric, i.e. $A=A^T$. We define the degree matrix $D=[d_{ii}] \in R^{N \times N}$ as a diagonal matrix as such \vspace{-1.5mm}
\begin{equation}
d_{ii}=\sum_{j=1}^{N} a_{ij} \vspace{-1.5mm}
\label{eq_rev1} 
\end{equation}
The matrix $L=[l_{ij}]=D-A$ is denoted as the Laplacian Matrix of $\mathcal{G}$. As mentioned above, $A=A^T$, and considering $D$ is a diagonal matrix, it follows that $L=L^T$. The neighbor set corresponding to each vertex $i$ is defined as $\mathcal{N}_{\;\;i}=\{j | (i,j) \in \mathcal{\varepsilon}\}$. Additionally, $\mathcal{N}_{\;\;j}^{+}$ denotes the set of outgoing neighbors of node $j$, i.e., the set of nodes receiving signals form the node $j$, and $\mathcal{N}_{\;\;j}^{-}$ is the set of nodes which sends signals to the node $j$. For the bidirectional graph $\mathcal{G}$, $\mathcal{N}_{\;\;j}^{+} = \mathcal{N}_{\;\;j}^{-}$. A graph is connected if there exists a path between any two distinct vertices \cite{dorfler2018electrical}. We assume that $\mathcal{G}$ is connected.

Next, consider \cref{Fig: Microgrid Schematic} which shows a sample localized microgrid with four DGs, denoted by DG$_i$, $i=1,2,3,4$. In this figure, the dashed lines show signaling between the cyber layer and physical layer, i.e. the communications between the DGs and their corresponding agents in the cyber layer. The lines with bidirectional arrows represent communications among the corresponding agents of DG$_i$ located in the cyber layer. The solid lines are electrical connections. Based on the weights shown in \cref{Fig: Microgrid Schematic} and the explanations above, the adjacency and degree matrices are defined as,

\small
\begin{equation} \label{eq: adjacency and diagonal}
\begin{aligned}
\!\!\!\!\!\!\!\!&A\!\!=\!\!\!\!\left[\begin{array} {cccc}  \!\!\!\! 0 & \!\!\!\! 0 & \!\!\!\! a_{13} & \!\!\!\! 0 \!\!\!\!\!\!\\ \!\!\!\! 0 & \!\!\!\! 0 & \!\!\!\! a_{23} & \!\!\!\! a_{24}\!\!\!\!\!\!\\ \!\!\!\!a_{31} & \!\!\!\! a_{32} & \!\!\!\! 0 & \!\!\!\! 0 \!\!\!\!\!\!\\ \!\!\!\! 0 \!\!\!\! & \!\!\!\!a_{42} & \!\!\!\! 0 \!\!\!\! & \!\!\!\! 0 \!\!\!\!\!\! \end{array} \right]\!\!\!, 
D\!\!=\!\!\!\!\left[\begin{array} {cccc} \!\!\!\! \! a_{13} & \!\!\!\!\!\!\!\! 0 & \!\!\!\!\!\!\!\! 0 & \!\!\!\!\!\!\!\! 0 \!\!\!\!\!\!\\ \!\!\!\!\! 0 & \!\!\!\!\!\!\!\! a_{23}+a_{24} & \!\!\!\!\!\!\!\! 0 & \!\!\!\!\!\!\!\! 0 \!\!\!\!\!\! \\ \!\!\!\!\!0 & \!\!\!\!\!\!\!\! 0 & \!\!\!\!\!\!\!\! a_{31}+a_{32} & \!\!\!\!\!\!\!\!0 \!\!\!\!\! \\ \!\!\!\!\! 0 & \!\!\!\!\!\!\!\! 0 & \!\!\!\!\!\!\!\! 0 & \!\!\!\!\!\!\!\! a_{42}\!\!\!\!\!\end{array}\right] \hspace{-7mm}
\end{aligned}
\end{equation}
\normalsize
Based on the definition of Laplacian matrix, the corresponding Laplacian matrix to the adjacency and diagonal matrices defined in (\ref{eq: adjacency and diagonal}) is 
\begin{equation}
\begin{aligned}
L=\left[\begin{array} {cccc} \!\!a_{13} &\!\! 0 & \!\! -a_{13} &\!\! 0 \\ \!\! 0 & \!\! a_{23} + a_{24} & \!\! -a_{23} & \!\! -a_{24} \\ \!\! -a_{31} & \!\! -a_{32} & \!\! a_{31} + a_{32} & \!\! 0 \\ 0 & \!\! -a_{42} & \!\! 0 & \!\! a_{42} \end{array}\right]
\end{aligned}   \vspace{-2mm}
\label{eq_rev2}
\end{equation}

\section{Problem Definition} \label{Sec: Power Defenition}
We consider a microgrid in the grid-connected mode, where the microgrid's voltage and frequency are imposed by the main grid, i.e. the microgrid's frequency and voltage are fixed. Hence, the goal in this mode is to control the output power of the DGs. The cyber-physical systems considered in this paper is similar to the one shown in \cref{Fig: Microgrid Schematic}. The proposed control emerges from consensus control of Multi-Agent Systems (MAS). The control objective is sharing load power in proportion to the maximum power capacity of the DGs, under variations in maximum capacities. 
\begin{figure}
	\begin{center} 
		\includegraphics[width=0.7\textwidth]{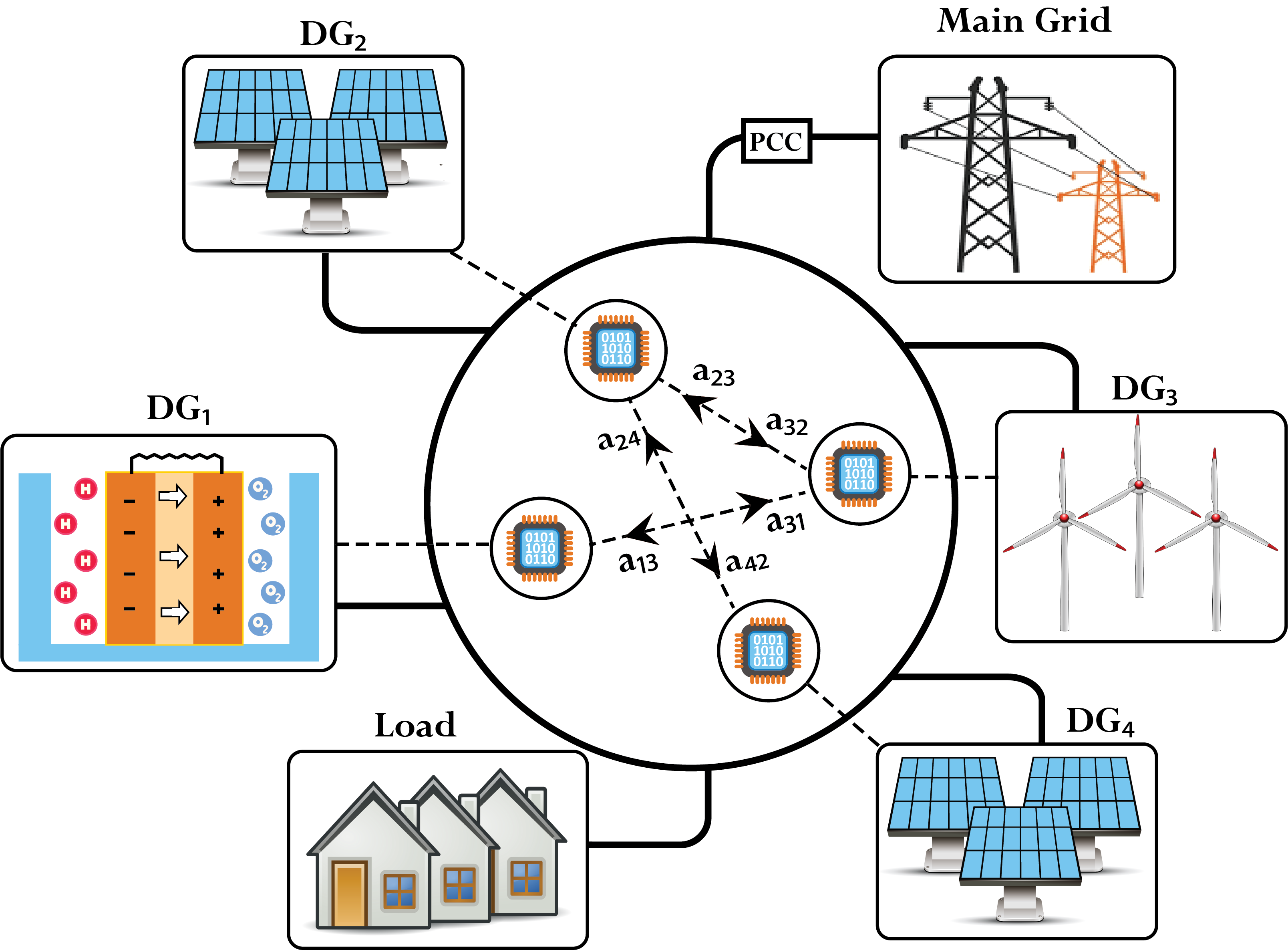}
	\end{center}
	\caption{Schematic of a microgrid comprising cyber and physical layers}
	\label{Fig: Microgrid Schematic} \vspace{-2mm}
\end{figure}
We assume that there exists $N$ DGs in a microgrid which are labeled as DG$_i$ where $i=1, 2,\cdots, N$. The maximum power capacity and instantaneous output power of each DG$_i$ are defined as $P_{i, max}$ and $P_i$, respectively. Let $P_L$ be the load power, which is proportionately shared among the DGs, i.e. \vspace{-1mm}
\begin{equation} \label{Eq: Power share ratio definition}
P_L=\sum_{i=1}^{N} P_{i}, \,\,\, \mbox{s.t.} \,\,\, r=\frac{P_{1}}{P_{1, max}}=\frac{P_{2}}{P_{2, max}}=\cdots=\frac{P_{N}}{P_{N, max}} \vspace{-1mm}
\end{equation} 
where $r$ is the proportional power share ratio. Thus, the output power of DG$_i$ is $P_i=rP_{i,max}$. Let $P_{T}$ be the total power capacity of the microgrid defined as the accumulation of the maximum power capacity of all the DGs in the microgrid. Then, one can conclude that 
\vspace{-0.5mm}
\begin{equation} \label{Eq: Proportional Ratio Root} 
r=\frac{\sum_{i=1}^{N} P_{i}}{\sum_{i=1}^{N} P_{i, max}}=\frac{P_L}{P_{T}} \vspace{-1mm}
\end{equation}
and the output power of DG$_i$ is $P_i=(P_L/P_T)\,P_{i,max}$. Note that a fluctuation in the maximum capacity of a DG $P_{i,max}$ or a change in $P_L$ will cause a change in $r$. The proposed power sharing control will, in response, manage the output power of the DGs flexibly. Throughout this study, the goal is to manage the variations of $P_{i}$s while meeting the requirement $P_L=\sum_{i=1}^{N} P_{i}$. It is assumed that all DGs have a knowledge of $P_L$ at all times. The power demand $P_L$ can vary with time.

We next explain two scenarios for which different controllers are designed. At the core of these controllers is a consensus algorithm which is inherently distributed. Recall that an underlying assumption is that the communication graph among the DGs is connected. Before any change happens to the renewable energy resources, we assume all DGs have the knowledge of $P_T$ by which they are able to compute $r$ from (\ref{Eq: Proportional Ratio Root}) and thereby generate their appropriate proportional power share $P_i=rP_{i,max}$, $i=1,2,\cdots,N$. 

In the first scenario, assume the maximum capacity of DG$_k$ which is $P_{k,max}$, changes. Then $P_T$ changes accordingly and all DGs are required to update their value of $P_T$ to be able to recalculate the new $r$ based on (\ref{Eq: Proportional Ratio Root}). The only DG that can generate accurate power immediately after a fluctuation happens is the DG$_k$ since it is aware of the change in $P_{k,max}$. Let $\delta$ be the change such that $\tilde{P}_{k,max}=P_{k,max}+\delta$, where $\tilde{P}_{k,max}$, is the updated value of $P_{k,max}$. Thus, DG$_k$ can compute the updated capacity of the microgrid as $\tilde{P}_{T}$ where  $\tilde{P}_{T}=P_{T}+\delta$ and recalculate $r$ and the delivered power $P_k$. A consensus algorithm is devised to have other DGs compute the $\tilde{P}_{T}$ and thereby reach the new value of $r$, distributively.

In the second scenario, we address the mismatch between load and supplied power before consensus is reached. As was discussed in scenario (1), only DG$_k$ can generate an accurate amount of power instantaneously after a fluctuation in DG$_k$. Although the other DGs are able to update $P_i$ following a change in $P_{k,max}$, the consensus algorithm takes time to converge, and hence during the transient time $\sum_{i=1}^{N} P_{i}$ would not necessarily be equal to $P_L$. The reason is that the other DGs do not have the correct value of $\tilde{P}_{T}$ instantaneously. However, since instantaneous matching of load power is a priority, a control law is augmented with the consensus algorithm to practically remove power mismatch during transients.  
\vspace{-10mm}
\section{Distributed Microgrid Control} \label{sec: sec4}

\subsection{Consensus on Total Power Capacity under Perturbation} \label{Sec: Static Consensus}
We consider a scenario where the individual DGs know the power ratio $r$ and generate accurate $P_i$ based on proportional power sharing, as shown in (\ref{Eq: Power share ratio definition}). Hence, each DG has correct knowledge of $P_T$, as per (\ref{Eq: Proportional Ratio Root}). Next, consider a change in $P_{k,max}$ to $\tilde{P}_{k,max}=P_{k,max}+\delta$. Following this change, all agents are required to compute $\tilde{P}_{T} = P_T + \delta$, the updated value of $P_T$. We define $s_i(t)$ as the estimate of $\tilde{P}_T$ by DG$_i$. The vector of estimate variables is then, $\mathbf{S}(t)=\left[\begin{array} {ccccc}s_1&s_2&\cdots&s_N \end{array} \right]^T$, where $N$ is the number of the DGs in the microgrid. As mentioned above, all the DGs know $P_T$ before any change happens. Therefore, the initial value of $\mathbf{S}$ is, $\mathbf{S}(0)=P_T \textbf{1}$ where $\textbf{1}^{N}=\left[\begin{array} {ccccc} 1&1&\cdots&1\end{array} \right]^T$. Thereafter, we propose the following consensus dynamics in the cyber layer, through which all DGs update their value of $P_T$ and converge to $\tilde{P}_{T}$.
\vspace{-0.5mm}
\begin{equation} 
\begin{aligned}  \label{Eq: Agents Communication Dynamics}
&\dot{s}_k(t)=\!-h\left(s_k(t)-\tilde{P}_{T}\right)-\!\!\!\! \sum_{j \in \mathcal{N}_{\;\;k}} \!\!\!\! a_{kj}\left(s_k(t)-s_j(t)\right) \!\!\\
&\dot{s}_i(t)=\!-\!\!\!\!\sum_{j \in \mathcal{N}_{\;\;i}} \!\!\!\! a_{ij}\left(s_i(t)-s_j(t)\right), \,\, i=\!1,2,...,N, \, i \neq k
\end{aligned} \vspace{-1mm}
\end{equation}
where $s_k(0)=P_{T}$ and $s_i(0)\!=P_{T}$. In (\ref{Eq: Agents Communication Dynamics}), $a_{ij}>0$ and it denotes the weight of the communication link between agents $i$ and $j$, where $i,j=1,2,\cdots,N$, $i \neq j$, and $h>0$ is a parameter chosen by the $k^{th}$ agent. The parameter $h$ represents a measure of the convergence rate of $s_i(t)$ towards $\tilde{P}_T$. Since the communication graph is bidirectional, therefore $a_{ij}=a_{ji}$, and this implies that the Laplacian matrix is symmetric, i.e. $L=L^T$ (see example in (\ref{eq_rev2}). From (\ref{Eq: Agents Communication Dynamics}), the following matrix equation is obtained
\begin{equation} \label{Eq: Matrix Format of the Agents Communication Dynamics}
\dot{\textbf{S}}=-(L+\Delta)\mathbf{S}+hd_k \tilde{P}_{T}
\end{equation}
where 
\begin{equation}
d_k=\!\! \left[\begin{array} {ccccccccc} \sigma_{1,k} \\ \sigma_{2,k} \\ \vdots \\ \sigma_{k,k} \\ \vdots \\ \sigma_{N,k}  \end{array} \right]\!\!\! = \!\! \left[\begin{array} {ccccccccc} 0 \\ 0 \\ \vdots \\ 1 \\ \vdots \\ 0  \end{array} \right]\!\!\! \in \! R^{N \times 1},
\Delta= h d_k d^T_k \in \!R^{N \times N}
\label{eq_rev3}
\end{equation}
In (\ref{eq_rev3}), $\sigma_{i,k}$ is the {\it Kronecker Delta} function, implying the $k^{th}$ element of $d_k$ is one and rest are zero. Also, $\Delta_{k, k} = h$ and all other elements are zero. We now propose and prove the following Lemma.
\begin{lemma}  \label{Lem: Lemma 1}
	The linear dynamic system defined in (\ref{Eq: Matrix Format of the Agents Communication Dynamics}) and (\ref{eq_rev3}) is input-to-state stable (ISS), and $S\rightarrow \tilde{P}_{T} \textbf{1}$ given the graph of communication among the agents is connected.
	\label{lem_rev1}
\end{lemma}

\begin{proof}
	The linear system of \eqref{Eq: Matrix Format of the Agents Communication Dynamics} and \eqref{eq_rev3} is ISS if $-(L+\Delta)$ is Hurwitz \cite{khalil2002nonlinear}. 
	The input is $\tilde{P}_{nom}$ which is constant and bounded, thus, if $-(L+\Delta)$ is Hurwitz the proof is complete. Since $L=L^T$, and by definition $\Delta=\Delta^T$, the matrix $-(L+\Delta)$ is symmetric. Hence, it is Hurwitz if $-(L+\Delta)<0$, i.e. negative definite. To prove this, it is required to show that for any vector $u \in R^{N}$, $u^T[-(L+\Delta)]u$ is strictly less than zero unless $u=0$. From \cref{eq_rev3}, 
	\begin{equation} \label{Eq: Negative definiteness test}
	u^T[-(L+\Delta)]u = -u^TLu - u^T \Delta u =-u^TLu- h{u_k}^2  
	\end{equation}
	where $h=\Delta_{kk}>0$, and $u_k$ is the $k^{th}$ element of the vector $u$. As the communication graph is connected, $L$ is positive semi-definite with a single zero eigenvalue \cite{mesbahi2010graph,dorfler2018electrical}, and it is diagonalizable \cite{horn1990matrix}, with all real eigenvalues. Let $\lambda_i$, $i=1,2,\cdots,N$ be the eigenvalues of $L$ in descending order, $\lambda_1 \geq \lambda_2 \geq \cdots>\lambda_{N-1} > \lambda_{N}=0$. The canonical form of $L$ is $L=V\Lambda V^T$, where $\Lambda$ is a diagonal matrix consisting of the eigenvalues of $L$, and $V$ is the right eigenvector-matrix,
	\begin{equation} \label{Eq: Diagonal Form of L}
	\Lambda \!\! =\!\!\left[\begin{array} {ccccc} \lambda_1&0&&0\\ 0&\lambda_2&&\\&&\ddots&\\0&&&\lambda_N \end{array} \right]\!\!, \,\, V\!\!=\left[\begin{array} {ccccc}\!\!\! v_1 & \!\!\! v_2 & \!\!\! \cdots & \!\!\! v_N \!\!\!\end{array} \right] 
	\end{equation}
	Since $v_N$ is the eigenvector corresponding to $\lambda_N = 0$, following the definition of $L$, $v_N=c \textbf{1}^N$ where $c \neq 0$ is a real value. Substituting $L=V\Lambda V^T$ into (\ref{Eq: Negative definiteness test}) and taking  (\ref{Eq: Diagonal Form of L}) into account, the following holds:
	\begin{equation} \label{Eq: Negative definiteness test 1} 
	\begin{aligned}
	-u^TLu-h{u_k}^2=& -z^T\!\! \left[\begin{array} {ccccc} \!\! \lambda_1 & \!\!\!\! 0 & \!\!\!\!& \!\!\!\! 0\!\!\\ \!\! 0 & \!\!\!\! \lambda_2 & \!\!\!\! & \!\!\\ \!\! & \!\!\!\! & \!\!\!\!\ddots & \!\! \\ \!\! 0 \!\!\!\! & \!\!\!\! & \!\!\!\! & \!\!\!\! \lambda_N \!\! \end{array} \right] z -h {u_k}^2 \\ 
	&=-\sum_{i=1}^{N} \lambda_i z_i^2 -h {u_k}^2
	\end{aligned}
	\end{equation} 
	where $z=V^Tu$. Since $V^T$ is a nonsingular matrix it is invertible, and its inverse matrix is $V$ and $u=Vz$. For any $z=\left[\begin{array} {ccccc} z_1&z_2&\cdots&z_N \end{array} \right]^T$ and $u_k$, the right hand side of (\ref{Eq: Negative definiteness test 1}) is negative, except $z_i=0 \; \forall \; i=1,2,\cdots,N-1$ and $u_k=0$. The remaining condition is $z=\left[\begin{array} {ccccc} 0&\cdots&0&z_N \end{array} \right]^T$. As $V$ is not singular, $u=Vz \neq 0$ while $z_N \neq 0$. According to (\ref{Eq: Negative definiteness test 1}) and since $v_N=c\textbf{1}^N$, 
	\begin{equation}
	u=Vz=\!\left[\begin{array} {ccccc}\!\!\!\!v_1 & \!\!\! v_2 & \!\!\! \cdots &\! \! \! c\textbf{1}\!\!\!\! \end{array} \right] \left[\begin{array} {ccccc} \!\!\!\! 0 & \!\!\! \cdots & \!\!\! 0 & \!\!\! z_N \!\!\!\!\end{array} \right]^T\!\!=cz_N \textbf{1} 
	\end{equation}
	Now that $u=cz_N \mathbf{1}$ and $c,z_N \in R$ are non-zero, $u_k=cz_N \ne 0$. However, this is in contradiction with the assumption made before which is $u_k=0$. Thus, (\ref{Eq: Negative definiteness test 1}) is negative for any vector $z$ and since $u=Vz$, (\ref{Eq: Negative definiteness test}) is negative for any $u \neq 0$ which proves $-(L+\Delta)$ is negative definite. We define $y=\mathbf{S}-\tilde{P}_{T}\textbf{1}$, therefore, $\mathbf{S}=y+\tilde{P}_{T}\textbf{1}$ and $\dot{\mathbf{S}}=\dot{y}$. By substituting $y$ and $\dot{y}$ into (\ref{Eq: Matrix Format of the Agents Communication Dynamics}), we obtain
	\begin{equation} \label{Eq: Simplified Consensus 1}
	\dot{y}=\!-(L+\Delta)y,\, y=\mathbf{S}-\tilde{P}_{T}\textbf{1},\, y(t_0)\!=\!P_{T}\textbf{1}-\tilde{P}_{T}\textbf{1}\!=\!-\delta \textbf{1} \!\!\!\!
	\end{equation}
	As $-(L+\Delta)$ is Hurwitz the dynamics of (\ref{Eq: Simplified Consensus 1}) is exponentially stable. It means $y \rightarrow 0$ and therefore $\mathbf{S} \rightarrow \tilde{P}_{T}\textbf{1}$. This completes the proof.
\end{proof}

We assume that at any time, only one DG can change its power capacity. Thereafter, any subsequent change in power capacity can be done once the consensus algorithm described above has converged. 

\subsection{Observations on Rate of Convergence}

From Weyl's theorem on eigenvalue inequalities for sum of two Hermitian matrices \cite{horn1990matrix}, the properties of Laplacian matrices in a connected network \cite{mesbahi2010graph}, and \cref{Lem: Lemma 1}, we have 
\begin{equation*}
\lambda_{N-1}(-L) \leq \lambda_N(-(L+\Delta)) < \lambda_N(-L)=0
\end{equation*}
assuming the eigenvalues of $-L$ are ordered as $\lambda_1(-L) \leq \lambda_2(-L) \leq \cdots\leq \lambda_{N-1} (-L)\leq \lambda_N(-L)=0$, and those of $-(L+\Delta)$ are ordered as $\lambda_1(-(L+\Delta)) \leq \lambda_2(-(L+\Delta)) \leq \cdots\leq \lambda_{N-1} (-(L+\Delta)) \leq \lambda_N(-(L+\Delta)) < 0$ for all $h > 0$. The eigenvalue $\lambda_N(-(L+\Delta))$ is the dominant eigenvalue of $-(L+\Delta)$ and hence it determines the rate of convergence of consensus. To explain the effect of $h$ on convergence rate, we provide the following lemma.

\begin{lemma}
	The rate of convergence of the consensus algorithm, given by (\ref{Eq: Matrix Format of the Agents Communication Dynamics}) and (\ref{eq_rev3}), increases with increasing values of $h > 0$.
\end{lemma}

\begin{proof}
	The characteristic equation of $-L$ can be expressed as $p(\lambda) = 0$. The roots of $p(\lambda)$, i.e. $\lambda_1(-L) \leq \lambda_2(-L) \leq \cdots\leq \lambda_{N-1} (-L)\leq \lambda_N(-L)=0$, are all real. It can be verified that the characteristic equation of $-(L + \Delta)$ takes the form $p(\lambda) + h q(\lambda)= 0$. Further, $p(\lambda)$ and $q(\lambda)$ are both monic polynomials, with $p(\lambda)$ and $q(\lambda)$ being polynomials of order $N$ and $(N-1)$ respectively. From \cref{lem_rev1}, we know that the roots of $p(\lambda) + h q(\lambda)$ satisfy $\lambda_1(-(L+\Delta)) \leq \lambda_2(-(L+\Delta)) \leq \cdots\leq \lambda_{N-1} (-(L+\Delta)) \leq \lambda_N(-(L+\Delta)) < 0$, $\forall \; h > 0$. Considering the characteristic equation of $-(L + \Delta)$ in root-locus form, i.e.
	\begin{equation*}
	1 + h \frac{q(\lambda)}{p(\lambda)} = 0, \quad h > 0,
	\end{equation*}
	we deduce that all roots of $q(\lambda)$ are real and negative. Otherwise, increasing $h$ will eventually cause some eigenvalues of $-(L + \Delta)$ to become complex conjugates and/or unstable, thereby contradicting \cref{lem_rev1}. This observation is in accordance with the rules of root locus. We further deduce that the largest root of $q(\lambda)$, namely $\lambda_{max}(q(\lambda))$, must satisfy
	\begin{equation*}
	\lambda_{N-1}(-L) \leq \lambda_{max}(q(\lambda)) < \lambda_N(-L)=0
	\end{equation*}
	Violating the above condition would also contradict \cref{lem_rev1}. Since $\lambda_N(-L)=0$ and $\lambda_{max}(q(\lambda)) < 0$ are the largest roots of $p(\lambda)$ and $q(\lambda)$ respectively, there is one branch of root locus that originates from $\lambda_N(-L)$ and ends at $\lambda_{max}(q(\lambda))$ for $0 \le h < \infty$. This branch is also the closest to the origin of all $N$ branches of the above root locus (which are all strictly along the negative real axis), and hence contains the locus of the dominant eigenvalue of $-(L + \Delta)$. With increasing $h$, this eigenvalue moves further to the left of the origin, thereby increasing the rate of convergence to the consensus proposed in \cref{lem_rev1}.  This completes the proof.
\end{proof}




\subsection{Proportional Power Sharing Strategies} \label{sec_rev2}
The consensus algorithm of \cref{Sec: Static Consensus} enable the DGs to compute the updated capacity of the microgrid under perturbation. In this section, we propose methods by which individual agents command power to the physical layer based on consensus. Subsequent to a capacity variation such as $\delta$ in \cref{Sec: Static Consensus}, three slightly different strategies are proposed through which the DGs meet the load demand $P_L$. The first and third strategies are discussed in detail. The second strategy is similar to the first and hence its details are omitted. Assuming at $t=t_0$, $P_{k,max}$ changes to $\tilde{P}_{k,max}=P_{k,max}+\delta$, the first strategy to generate $P_i$s is
\begin{equation} \label{Eq: Strategy 1}
\mbox{Strategy \!\!\! 1:} \left\{ \begin{aligned}
&\!P_k=\!\frac{P_L}{s_k} \tilde{P}_{k,max} \\
&\!P_i=\!\frac{P_L}{s_i} P_{i,max} \,\,\,\mbox{for} \,\,\; i=1,2,\cdots,N \; i\neq k \!\!\!\!
\end{aligned} \right.
\end{equation}
where $s_i$, $i=1,2,\cdots,N$, are the estimates of $\tilde{P}_T$, as discussed in \cref{Sec: Static Consensus}, and $\lim_{t \to \infty} s_i = \tilde{P}_T$ according to \cref{lem_rev1}. A potential issue may arise when $s_i(t)$ crosses or approaches zero for some $t>t_0$ such that $P_i$ diverges. In this regard, we state and prove the following Lemma: 
\begin{lemma} \label{lem: Lemma2}
	Considering the LTI system defined in (\ref{Eq: Matrix Format of the Agents Communication Dynamics}), if $|\delta|<{\theta P_T}/{(1+\sqrt{N}})$ with $0 < \theta < 1- (P_L/P_T)$, the following holds
	\begin{equation} \label{Eq:limits of yi+tildeP}
	(1\!-\theta) P_T\!\! \leq s_i(t) \! \leq \!\! (1\!+\theta) P_T, \,\, \,\, P_i(t)\!\!<\!P_{i,max}, \,\,\, \forall \, t \!> t_0 \!\!\!
	\vspace{0.1in}
	\end{equation}
\end{lemma}
The proof of \cref{lem: Lemma2} is given in the Appendix. From \cref{lem: Lemma2}, it may appear that as the number of DGs, $N$, increases, there will be a bigger restriction on $\delta$, since $|\delta|<{\theta P_T}/{(1+\sqrt{N}})$. However, it can be shown that the above inequality is not restrictive, mainly because as $N$ increases, $P_T$ also increases. An analysis of this aspect is given in the Appendix. From \cref{lem: Lemma2}, it may also appear that the constraint on $\delta$ is restrictive as $\theta \to 0$, which happens as $P_L \to P_T$. This restriction is however justified, since $P_L \approx P_T$ practically implies that the grid is already close to maximum capacity. Hence, further perturbation in DGs may prevent it from meeting the load demand. So far, it is proved that strategy 1 is valid provided changes in $\delta$ satisfy the conditions in \cref{lem: Lemma2}. Defining the total instantaneous output power of the microgrid as $P_O(t)$, from (\ref{Eq: Strategy 1}), 
\begin{equation}
P_O(t) =\frac{P_L}{s_k} \tilde{P}_{k,max}+\sum_{i=1, i\neq k}^{N} \frac{P_L}{s_i}P_{i,max} 
\end{equation}
Therefore,
\begin{equation}
P_O(t) =\frac{P_L}{s_k} \delta+\sum_{i=1}^{N} \frac{P_L}{s_i}P_{i,max} 
\end{equation}
Thus, defining the instantaneous error $E(t) = P_O(t)-P_L$, we have,
\begin{equation} \label{Eq: Error between PO and PL}
E(t)=P_O(t)-P_L=\frac{P_L}{s_k} \delta+\sum_{i=1}^{N} \frac{P_L}{s_i}P_{i,max}-P_L 
\end{equation}
At $t=t_0$, $s_i=P_T$ for $i=1,2,\cdots,N$. Thus,
\begin{equation}
E(t_0)=P_L\Bigl[ \frac{\delta}{P_T}+\sum_{i=1}^{N} \frac{P_{i,max}}{P_T}-1 \Bigr]
\end{equation}
Since $\sum_{i=1}^{N}\frac{P_{i,max}}{P_T}=1$, therefore
\begin{equation} \label{eq_rev10}
E(t_0)=P_L\frac{\delta}{P_T}
\end{equation}
Equation (\ref{eq_rev10}) shows that $E(t_0) \neq 0$, and since $E(t)$ is continuous, it implies that a perturbation $\delta$ causes a transient mismatch between the delivered power $P_O(t)$ and the load $P_L$. The error $E(t) \to 0$ at steady-state, as proven in \cref{Lem: Lemma 1}. Therefore, Strategy 1 given in (\ref{Eq: Strategy 1}), causes a temporary mismatch of power following a perturbation. This issue is addressed in \cref{Sec: Dynamic Consensus}.

The second strategy, which is slightly different from the first one, is as follows: 
\begin{equation} \label{Eq: strategy 2 modified}
\mbox{ Strategy \!\!\! 2:} \left\{ \begin{aligned}
&\! P_k=\!\frac{P_L}{\tilde{P}_{T}}\tilde{P}_{k,max}\\
&\!P_i =\!\frac{P_L}{s_i}P_{i,max}\,\,\,\mbox{for} \,\,\; i=1,2,\cdots,N \; i\neq k \!\!\!\!\!
\end{aligned} \right.
\end{equation}
As before, the total instantaneous output power of the microgrid $P_O(t)$ is
\begin{equation} \label{Eq: output power strategy 1}
P_O(t)=\frac{P_L}{P_T+\delta}(\tilde{P}_{k,max})+\sum_{i=1, i\neq k}^{N} \frac{P_L}{s_i}P_{i,max} 
\end{equation}
We again evaluate the error $E(t)=P_O(t)-P_L$ for $t \geq t_0$, yielding
\begin{equation} 
E(t)=\frac{P_L}{P_T+\delta}(P_{k,max}+\delta)+\sum_{i=1, i\neq k}^{N} \frac{P_L}{s_i}P_{i,max} -P_L
\end{equation}   
Upon simplifying, we obtain
\begin{equation*} \label{Eq: Error between $P_O$ and $P_L$}
E(t)=P_L\Bigl[\frac{(P_{k,max}+\delta)}{P_T+\delta}+\sum_{i=1, i\neq k}^{N} \frac{P_{i,max}}{s_i}-1\Bigr]
\end{equation*}
Since at $t=t_0$, $s_i=P_T$ for $i=1,2,\cdots,N$,
and $\sum_{i=1, i \neq k}^{N} P_{i,max}=P_T-P_{k,max}$, 
\begin{equation} \label{Eq: Final E}
E(t_0)=P_L\frac{\delta (P_T-P_{k,max})}{P_T(P_T+\delta)}
\end{equation}
Equation (\ref{Eq: Final E}) shows that $E(t_0) \neq 0$, and since $E(t)$ is continuous, it implies that similar to Strategy 1, a perturbation $\delta$ causes a transient mismatch between the delivered power $P_O(t)$ and the load $P_L$ in Strategy 2. The error $E(t) \to 0$ at steady-state, as proven in \cref{Lem: Lemma 1}.

The last candidate strategy is proposed as 
\begin{equation} \label{Eq: strategy 3} 
\mbox{ Strategy \!\!\! 3:} \left\{ \begin{aligned}
&\!P_k=\!\frac{P_L}{s_k}({P_{k,max}}+s_k-P_T) \\
&\!P_i=\!\frac{P_L}{s_i}P_{i,max} \,\,\, \mbox{for} \,\,\; i=1,2,\cdots,N \; i\neq k \!\!\!\!\!
\end{aligned} \right.
\end{equation}
The Strategy 3 allows DGs to update their output power more smoothly compared to the first two strategies.
In this case, 
\begin{equation}
P_O(t)=\frac{P_L}{s_k}(P_{k,max}+s_k-P_T)+\sum_{i=1, i\neq k}^{N} \frac{P_L}{s_i}P_{i,max} 
\end{equation}
Therefore,
\begin{equation} \label{Eq: Error of the third strategy}
E(t)=P_L\Bigl[-\frac{P_T}{s_k}+\sum_{i=1}^{N} \frac{P_{i,max}}{s_i}\Bigl]
\end{equation}
Since $s_i=P_T$ for all $i=1, 2,\cdots, N$, at $t=t_0$, $E(t_0)=0$. However, $E(t)$ still undergoes transient fluctuations. Based on (\ref{Eq: Error of the third strategy}),
\begin{equation} \label{Eq: Edot simplified}
\frac{\dot{E}(t)}{P_L}= \frac{P_T \dot{s}_k(t)}{s_k^2(t)}-\sum_{i=1}^{N} \frac{P_{i,max} \dot{s}_i(t)}{s_i^2(t)}
\end{equation}
Equation (\ref{Eq: Edot simplified}) can be further simplified using (\ref{Eq: Matrix Format of the Agents Communication Dynamics}) as follows,
\begin{equation} \label{Eq: Expanded Derivative of E}
\begin{aligned}
\frac{\dot{E}(t)}{P_L}= &\frac{P_T \dot{s}_k}{s_k^2}- \\
&\left[\begin{array} {ccccccc} \frac{P_{1,max}}{s_1^2}&\frac{P_{2,max}}{s_2^2}&\cdots&\frac{P_{N,max}}{s_N^2}\end{array} \right] \\ &\hspace{15mm}\times[-(L+\Delta)S+hd_k(P_T+\delta)]
\end{aligned}
\end{equation}
Since at $t=t_0$, $s_i=P_T$ for all $i=1,2,\cdots,N$, $\mathbf{S}(t_0)=P_T\mathbf{1}$ and (\ref{Eq: Expanded Derivative of E})
becomes 
\begin{equation}\label{Eq: Derivative of E}
\begin{aligned}
\frac{\dot{E}(t_0)}{P_L}=&\frac{P_T \dot{s}_k(t_0)}{P_T^2} - \\ 
&\left[\begin{array} {ccccccc} \!\! \frac{P_{1,max}}{P_T^2} & \!\!\! \frac{P_{2,max}}{P_T^2} & \!\!\! \cdots & \!\!\! \frac{P_{N,max}}{P_T^2}\end{array} \right] \\
&\hspace{8mm}\times\left[-L\mathbf{1}P_T-hd_kP_T+hd_kP_T+hd_k\delta \right]
\end{aligned}
\end{equation}
Simplifying (\ref{Eq: Derivative of E}) yields
\begin{equation} \label{Eq: final Edot(t0)}
\frac{\dot{E}(t_0)}{P_L}=\frac{P_T \dot{s}_k(t_0)}{P_T^2}-\frac{hP_{k,max} \delta}{P_T^2}
\end{equation} 
Referring to (\ref{Eq: Matrix Format of the Agents Communication Dynamics}), at $t=t_0$, the term $\dot{s}_k(t_0)$ in (\ref{Eq: final Edot(t0)}) is 
\begin{equation} \label{Eq: skdot at t0}
\dot{s}_k(t_0)=h(P_T+\delta)-hP_T=h\delta
\end{equation}
Therefore, from (\ref{Eq: final Edot(t0)}) and (\ref{Eq: skdot at t0}),
\begin{equation}
\frac{\dot{E}(t_0)}{P_L}=h\delta \frac{P_T-P_{k,max}}{P_T^2} 
\end{equation}
where from (\ref{Eq: Agents Communication Dynamics}), $h$ is a positive scalar chosen by $k^{th}$ agent. Thus, although $E(t_0)=0$, $\dot{E}(t_0) \neq 0$.
Therefore, as $E(t)$ is continuous, similar to Strategies 1 and 2, a change $\delta$ results in a transient mismatch between $P_O$ and $P_L$. It is shown that the three strategies proposed above match the load power $P_L$ at steady-state while producing transient deviations. This transient issue is resolved in the next section, where a strategy is proposed to practically maintain $P_O=P_L$ at any time. \vspace{-2mm}

\subsection{Proportional Power Sharing with Transient Power Match} \label{Sec: Dynamic Consensus}
Upon a perturbation in $P_{k,max}$, which results in a change in $P_T$, the agents estimate $\tilde{P}_T=P_T+\delta$ through (\ref{Eq: Matrix Format of the Agents Communication Dynamics}). Among the DGs, only DG$_k$ has a knowledge of $\tilde{P}_T=P_T+\delta$. The other DGs in the microgrid converge to $\tilde{P}_T$ through consensus only at steady-state. This leads to the transient power mismatch discussed in \cref{sec_rev2}. To remove this transient mismatch, we propose a strategy where DG$_k$ modulates its power delivery as follows, while the other DGs maintain the same strategy as in \cref{sec_rev2}: 
\begin{equation} \label{Eq: Keeping Power Sharing Control} 
\begin{aligned}
&P_k=\frac{P_L}{s_k} P'_{k,max} \\
&P_i=\frac{P_{L}}{s_i}P_{i,max} \qquad \quad \mbox{for} \quad i=1,2,\cdots,N \quad i\neq k
\end{aligned}
\end{equation}
where $P'_{k,max}$ is an auxiliary dynamic variable required to modulate the instantaneous power of DG$_k$. Hence, at $t=t_0$, $P'_{k,max}(t_0)=P_{k,max}(t_0)$, and it is required to converge to $(P_{k,max}+\delta)$ while $s_i$ converges via consensus.  
With the goal of maintaining $P_O(t)=P_L$ for all $t > t_0$, we must have
\begin{equation} \label{Eq: the output power equals load} 
P_L=P_O(t)=\frac{P_L}{s_k} P'_{k,max}+\sum_{i=1,i\neq k}^{N}\frac{P_{L}}{s_i}P_{i,max}
\end{equation}
Thus,
\begin{equation} \label{Eq: Derivative of PO-PL}
\frac{P'_{k,max}}{s_k} +\sum_{i=1,i\neq k}^{N}\frac{P_{i,max}}{s_i}-1=0
\end{equation}
Therefore, $P'_{k,max}$ is 
\begin{equation}  \label{Eq: update of Pkmaxprime}
P'_{k,max}(t) = s_k(t) \Bigl[ 1 - \sum_{i= 1, i\neq k}^{N} \frac{P_{i,max}}{s_i(t)} \Bigr]    
\end{equation}
The algorithm for updating $P'_{k,max}$ and $s_i$ for $i=1, 2, \cdots, N$ in (\ref{Eq: Keeping Power Sharing Control}) is as follows:
\begin{subequations}  \label{eq: proposed algorithm} 
	\begin{align}
	&P'_{k,max}(t)= s_k(t) \Bigl[1 - \sum_{i= 1, i\neq k}^{N} \frac{P_{i,max}}{s_i(t)}\Bigr] \label{Eq: DG_K strategy 1}\\
	&\dot{\textbf{S}}=-(L+\Delta)\mathbf{S}+hd_k \tilde{P}_{T} \quad \,\,\, \mbox{where} \,\,\,\,\,\,\, \mathbf{S}(t_0)=P_T \mathbf{1} \label{Eq: update of Si in realtime}
	\end{align}
\end{subequations}
Based on (\ref{Eq: Keeping Power Sharing Control}) and (\ref{eq: proposed algorithm}), we state and prove the following lemma:
\begin{lemma} \label{Lem: 2}
	The dynamic system of (\ref{eq: proposed algorithm}) is stable, i.e. the terms $P'_{k,max}$, $\frac{P_{i,max}}{s_i}$ and $\textbf{S}$ remain bounded if $|\delta|<{\theta P_T}/{(1+\sqrt{N}})$, where $0 < \theta < 1- (P_L/P_T)$. Furthermore, $P'_{k,max} \rightarrow \tilde{P}_{k,max}$ and $\mathbf{S} \rightarrow \tilde{P}_T \mathbf{1}$, while the instantaneous delivered power satisfies (\ref{Eq: the output power equals load}) for all $t \geq t_0$.
\end{lemma}

\begin{proof}
	Since (\ref{Eq: update of Si in realtime}) is equivalent to (\ref{Eq: Matrix Format of the Agents Communication Dynamics}), per \cref{Lem: Lemma 1}, the dynamic system of (\ref{Eq: update of Si in realtime}) is ISS. Therefore $\textbf{S}$ is bounded. Additionally, as (\ref{Eq: update of Si in realtime}) and \ref{Eq: Matrix Format of the Agents Communication Dynamics}) have the same initial conditions, i.e. $\mathbf{S}(t_0)=P_T \mathbf{1}$, thus $\mathbf{S} \rightarrow \tilde{P}_T \mathbf{1}$. 
	Following $|\delta|<{\theta P_T}/{(1+\sqrt{N}})$, from \cref{lem: Lemma2} we have $(1-\theta) P_T \leq s_i(t) \leq (1+\theta) P_T$ with $0 < \theta < 1- (P_L/P_T)$. Thus,
	\begin{equation}
	\frac{P_{i,max}}{(1+\theta)P_T} \leq \frac{P_{i,max}}{s_i} \leq \frac{P_{i,max}}{(1-\theta)P_T}   
	\end{equation}
	Therefore, $\frac{P_{i,max}}{s_i}$ is bounded for all $i=1, 2, \cdots, N$. It demonstrates that (\ref{Eq: DG_K strategy 1}) represents a viable way to update $P'_{k,max}$. By plugging $P'_{k,max}$ from (\ref{Eq: DG_K strategy 1}) into (\ref{Eq: Keeping Power Sharing Control}), $P_O(t)$ simplifies to
	\begin{equation}
	P_O(t)=\frac{P_L}{s_k} s_k \Bigl[ 1 \!-\!\!\!\! \sum_{i= 1, i\neq k}^{N} \!\!\!\! \frac{P_{i,max}}{s_i} \Bigr] \!+ \!\!\!\!\! \!\sum_{i=1,i\neq  k}^{N}\!\!\frac{P_{L}}{s_i}P_{i,max} = P_L
	\end{equation}
	for all $t\geq t_0$. Since $\mathbf{S}$ converges to $\tilde{P}_T \mathbf{1}$, from (\ref{Eq: DG_K strategy 1}) we therefore deduce
	\begin{equation}
	P'_{k,max}(t) \rightarrow \tilde{P}_T \Bigl[ 1 - \sum_{i= 1, i\neq k}^{N} \frac{P_{i,max}}{\tilde{P}_T} \Bigr] = \tilde{P}_{k,max} 
	\end{equation}
	This completes the proof. 
\end{proof}
The controller designed in (\ref{Eq: DG_K strategy 1}) and (\ref{Eq: update of Si in realtime}) maintains $P_O(t)=P_L$ following a variation in the power capacity of a DG, namely $P_{k,max}$. However, to compute the term 
\begin{equation} \label{Eq:  Accumulation of Terms except pk/sk}
\sum_{i=1,i\neq k}^{N}\frac{P_{i,max}}{s_i}
\end{equation}
in $P'_{k,max}$, as given in (\ref{Eq: DG_K strategy 1}), the $k^{th}$ agent requires additional information. The following approach is proposed to enable the $k^{th}$ agent to attain this information distributively. This approach is based on the distributed finite-time average consensus studied in \cite{charalambous2015distributed}. According to \cite{charalambous2015distributed}, each agent $i$, shares $\frac{P_{i,max}}{s_i(t)}$ to its outgoing neighbors $\mathcal{N}_{\;\;i}^{+}$ where, following \cref{Sec: Preliminary Definitions}, $\mathcal{N}_{\;\;i}^{+}$ stands for the set of nodes which receives signals from node $i$. Accordingly, based on what follows the agents are able to distributively compute the instantaneous average of all $\frac{P_{i,max}}{s_i(t)}$ where $i=1, 2, \cdots, N$, i.e.
\begin{equation} \label{Eq: Average}
C_a(t)=\frac{\sum_{i=1}^{N}\frac{P_{i,max}}{s_i(t)}}{N}
\end{equation} 
Then, the $k^{th}$ agent can compute (\ref{Eq:  Accumulation of Terms except pk/sk}) via 
\begin{equation} \label{Eq: Computing the accumulation term}
\sum_{i=1,i\neq k}^{N}\frac{P_{i,max}}{s_i}=N \, C_a(t)-\frac{P_{k,max}}{s_k}
\end{equation}
One example of applying this distributed finite-time average consensus is presented in \cite{aalipour2017distributed}. Similar to \cite{aalipour2017distributed}, the steps of executing the finite-time algorithm is as following:
\begin{equation} 
\begin{aligned}  
\overline{g}_i(m+1)&=p_{ii}\overline{g}_i(m)+\sum_{j\in\mathcal{N}_{\;\;i}^{-}}p_{ij}\overline{g}_j(m)\\
g_i(m+1)&=p_{ii}g_i(m)+\sum_{j\in\mathcal{N}_{\;\;i}^{-}}p_{ij}g_j(m)
\end{aligned}\label{finite}
\end{equation}
where $\overline{g}_i(0)=\frac{P_{i,max}}{s_i}$ and $g_i(0)=1$ for $i=1, 2, \cdots, N$. Additionally, $p_{ij}=\frac{1}{1+|\mathcal{N}_{\;\;j}^{+}|}$ for $i\in\mathcal{N}_{\;\;j}^{+}\cup\{j\}$, otherwise is zero. Let us define the vectors 
\small
\begin{equation} \label{Eq: Vectors of successive differences}
\begin{aligned} 
\hspace{-11mm}&\overline{g}_{i,2m}^T\!\!=\!\![\overline{g}_i(1)\!-\!\overline{g}_i(0), \overline{g}_i(2)\!-\!\overline{g}_i(1),\cdots, \overline{g}_i(2m+1)\!-\!\overline{g}_i(2m)] \hspace{-10mm}\\
\hspace{-11mm}&g_{i,2m}^T\!\!=\!\![g_i(1)\!-\!g_i(0), g_i(2)\!-\!g_i(1), \cdots, g_i(2m+1)\!-\!g_i(2m)]  \hspace{-7mm}
\end{aligned} 
\end{equation}
\normalsize
and the following Hankel matrices
\begin{equation} \label{Eq: Hankel 1} 
\small{
	\Gamma\{\overline{g}_{i,2m}^T\}\triangleq
	\begin{bmatrix}\overline{g}_{i,2m}(1) & \cdots & \overline{g}_{i,2m}(m+1)\\
	\overline{g}_{i,2m}(2) & \cdots &\overline{g}_{i,2m}(m+2)\\
	\vdots & \ddots &\vdots\\
	\overline{g}_{i,2m}(m+1) &\cdots & \overline{g}_{i,2m}(2m+1)
	\end{bmatrix}}\normalsize
\end{equation}
and
\begin{equation} \label{Eq: Hankel 2} 
\small{
	\Gamma\{g_{i,2m}^T\}\triangleq
	\begin{bmatrix}g_{i,2m}(1) & \cdots & g_{i,2m}(m+1)\\
	g_{i,2m}(2) & \cdots & g_{i,2m}(m+2)\\
	\vdots & \ddots &\vdots\\
	g_{i,2m}(m+1) &\cdots & g_{i,2m}(2m+1)
	\end{bmatrix}}\normalsize
\end{equation}
Each agent $i$ runs the steps in (\ref{finite}) for $2N+1$ times and keeps the values $\overline{g}_i(m)$ and $g_i(m)$ for $m=1, 2, \cdots, 2N+1$. Having $\overline{g}_i(m)$ stored for the $2N+1$, each agent $i$ establishes the vectors $\overline{g}_{i,2m}^T$ and $g_{i,2m}^T$ defined in (\ref{Eq: Vectors of successive differences}) starting from $m=0$. At the same time, all individual agents construct their Hankel matrices $\Gamma\{\overline{g}_{i,2m}^T\}$ and $\Gamma\{g_{i,2m}^T\}$ defined in (\ref{Eq: Hankel 1}) and (\ref{Eq: Hankel 2}), respectively. Additionally, they calculate the ranks of the Hankel matrices for each $m$ and repeat the same procedure for the next $m+1$ until for a specific $m$ either $\Gamma\{\overline{g}_{i,2m}^T\}$ or $\Gamma\{g_{i,2m}^T\}$ becomes a defective matrix. Assume $\Gamma\{\overline{g}_{i,2M_i}^T\}$ or $\Gamma\{g_{i,2M_i}^T\}$ is the first matrix which loses its full rank where $\beta_i=[\beta_{i,0},\cdots,\beta_{i,M_i-1},1]^T$ is its corresponding kernel. Having the kernel $\beta_i$, the $i^{th}$ agent computes the average of all $\overline{g}_i(0)=\frac{P_{i,max}}{s_i}$ for $i=1, 2, \cdots, N$ defined as $C_a$ in (\ref{Eq: Average}) through the following
\begin{equation} \label{Eq: final average finite step}
C_a(t) = \frac{1}{N} \sum_{i=1}^{N}\overline{g_i}(0)=\frac{\left[\overline{g}_{i}(0),\overline{g}_{i}(1),\cdots,\overline{g}_{i}(M_i)\right]\beta_i}{[g_i(0),g_i(1),\cdots,g_i(M_i)]\beta_i}
\end{equation}
Thereby, the $k^{th}$ agent achieves $C_a(t)$, distributively. At this step, the $k^{th}$ agent obtains the term in (\ref{Eq:  Accumulation of Terms except pk/sk}) via (\ref{Eq: Computing the accumulation term}).
By plugging (\ref{Eq:  Accumulation of Terms except pk/sk}) back to (\ref{Eq: update of Pkmaxprime}) the $k^{th}$ agent is able to compute $P'_{k,max}$. 

To implement the proposed strategy practically, the mentioned procedures are required to be discretized firstly, since in practice the signals and algorithms update, 
digitally. We define an index $w$, starting from $w=0$, that represents the discrete instants at which the overall consensus algorithm is executed. At $w=0$, the $i^{th}$ agent, $i=1,2,\cdots,N$, has the value of $s_i(w)=P_{T}$. Therefore, they can compute $P_{i,max}/s_i(w)$, individually. Following the finite-time algorithm, they implement the procedure in (\ref{finite})-(\ref{Eq: final average finite step}). 
Now that all the agents, including the $k^{th}$ agent, obtains $C_a(w)$ from (\ref{Eq: final average finite step}), then using $C_a(w)$, the $k^{th}$ agent computes
\begin{equation}  
P'_{k,max}(w) = s_k(w) \Bigl[ 1 - \sum_{i= 1, i\neq k}^{N} \frac{P_{i,max}}{s_i(w)} \Bigr]    
\end{equation}
The command signals to DGs are as follows
\begin{equation}  
\begin{aligned}
&P_k(w)=\frac{P_L(w)}{s_k(w)} P'_{k,max}(w) \\
&P_i(w)=\frac{P_{L}(w)}{s_i(w)}P_{i,max} \qquad \quad \mbox{for} \quad i=1,2,\cdots,N \quad i\neq k
\end{aligned}
\end{equation}
Thereafter, the agents compute $s_i(w+1)$ through
\begin{equation} \label{Eq: discretized update of consenus}
\textbf{S}(w+1) = \textbf{S}(w)+dt[-(L+\Delta)\mathbf{S}(w)+hd_k \tilde{P}_{T}]    
\end{equation}
Set $w=w+1$, and repeat the above strategy with the updated value of the $s_i(w)$ defined in (\ref{Eq: discretized update of consenus}).
\begin{figure}[!htp]
	\centering
	\subfloat[]
	{
   \includegraphics[width=.9 \textwidth]{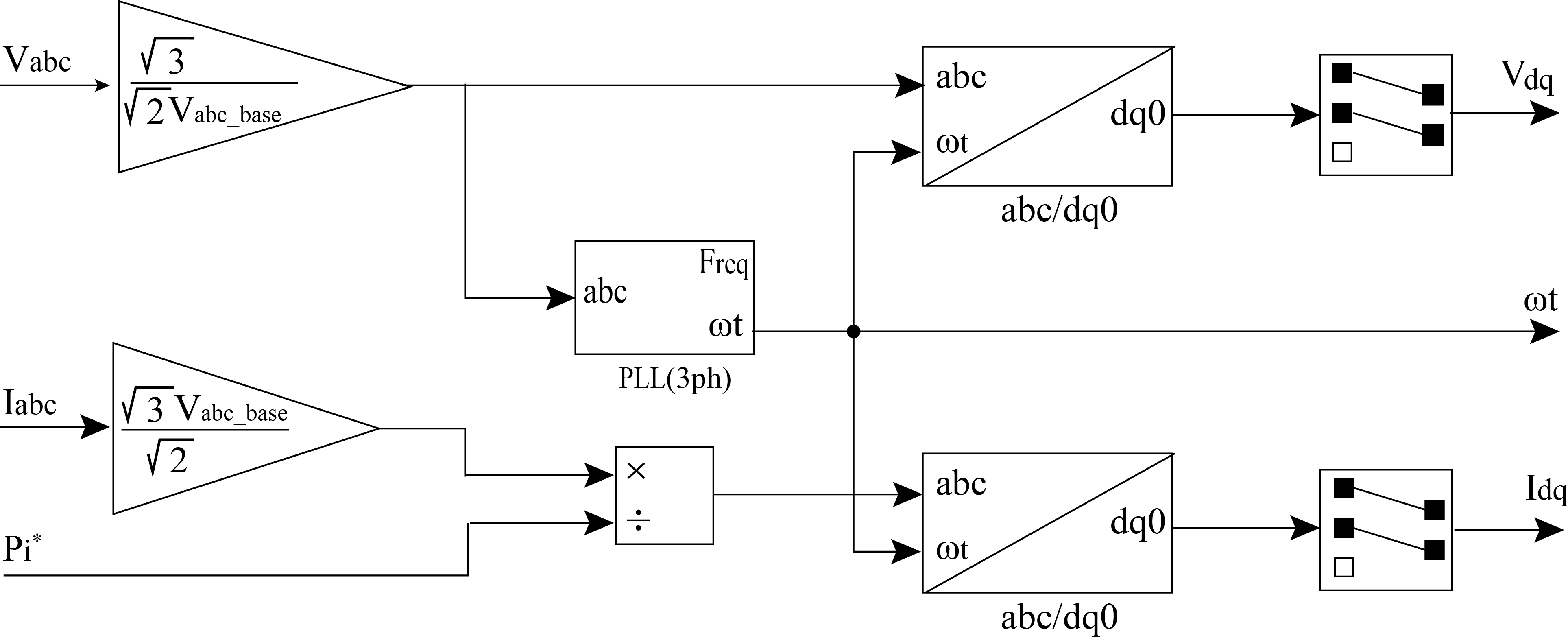}
		\label{fig: Power control and abc to dq0}
	}
	\hfill
	\subfloat[]
	{ 
   \includegraphics[width=.9 \textwidth]{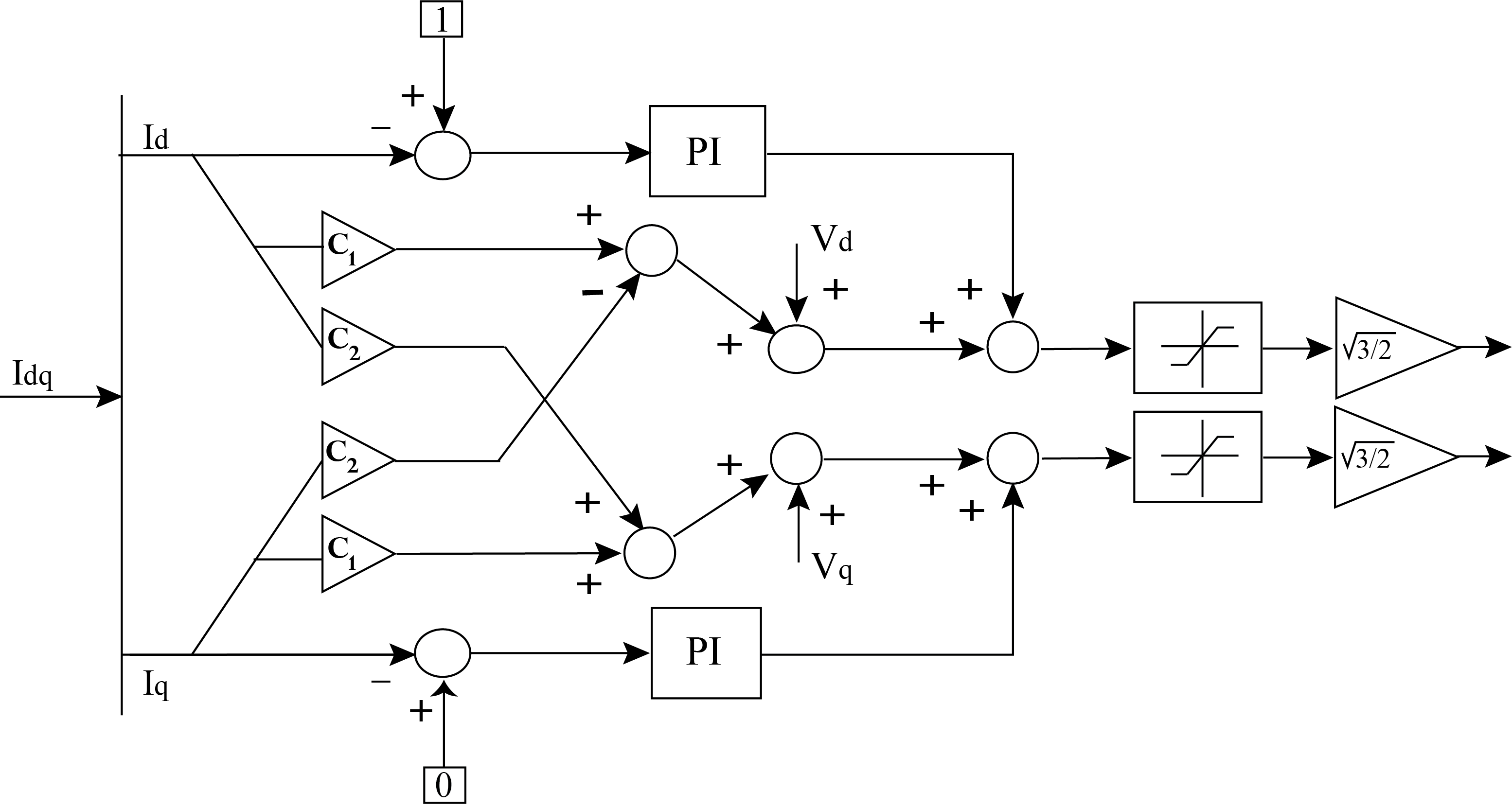}
		\label{fig: Current Control}
	}
	\hfill
	\subfloat[]
	{
   \includegraphics[width=.9 \textwidth]{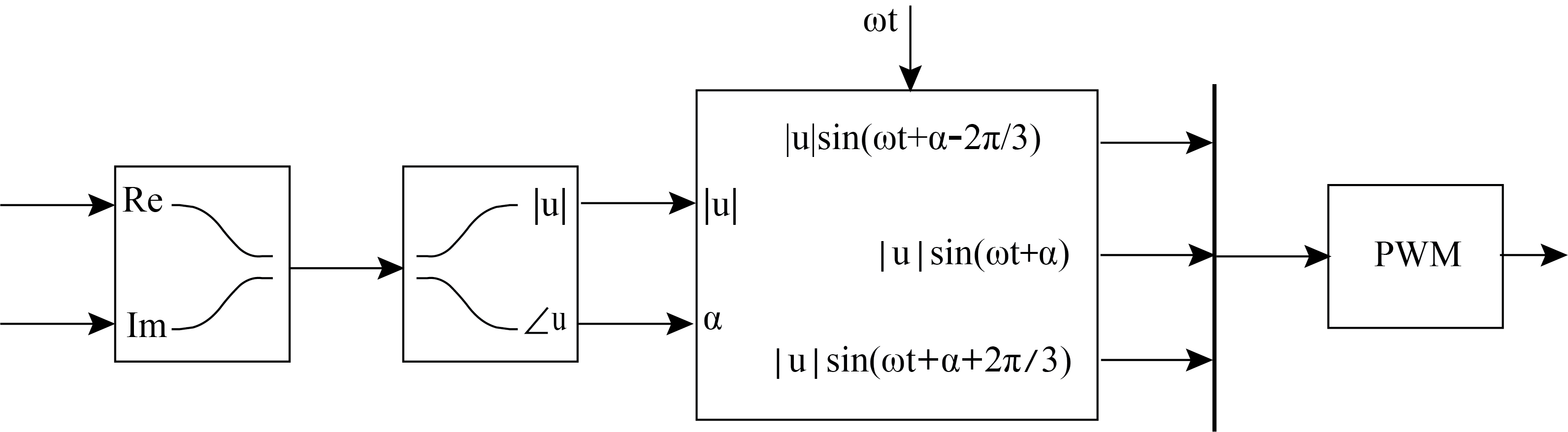}
		\label{fig: Voltage Control}
	} 
	\caption{Physical layer control scheme}
	\label{fig: Whole Power Control Diagram}
\end{figure}
\begin{figure}[t!]
	\begin{center} 
\includegraphics[width=.9\textwidth]{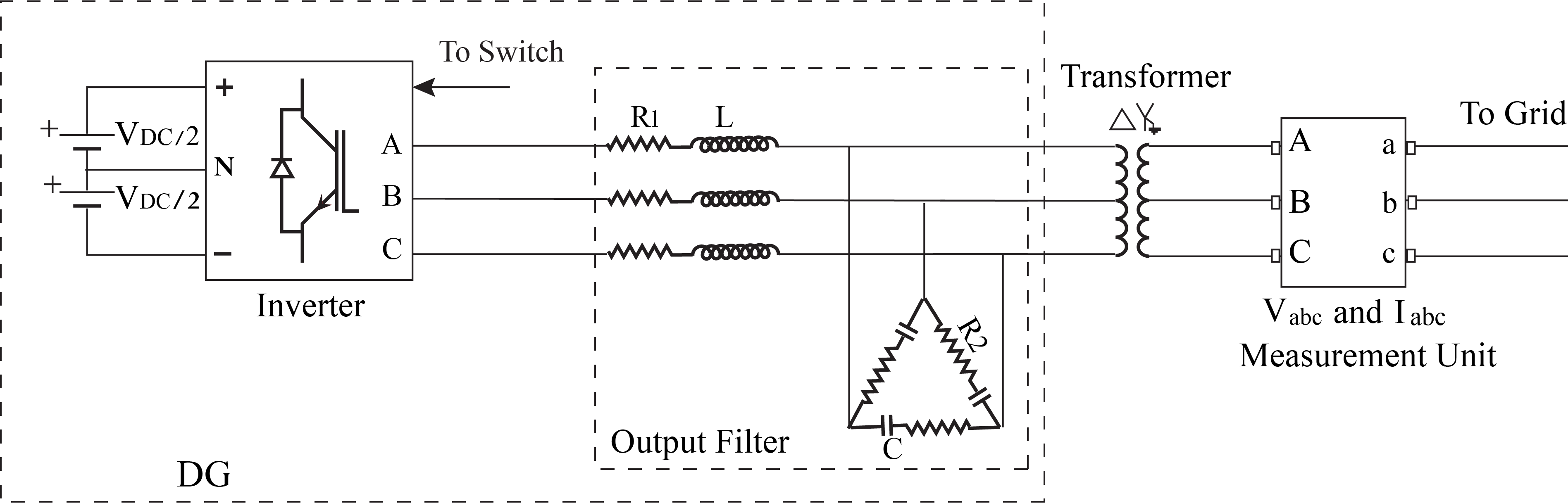}
	\end{center}
	\caption{Power circuit diagram of a DG}
	\label{fig: DG Power Circuit Diagram} \vspace{-4mm}
\end{figure}
\begin{figure}[t!]
	\centering
	\includegraphics[width=.9 \textwidth]{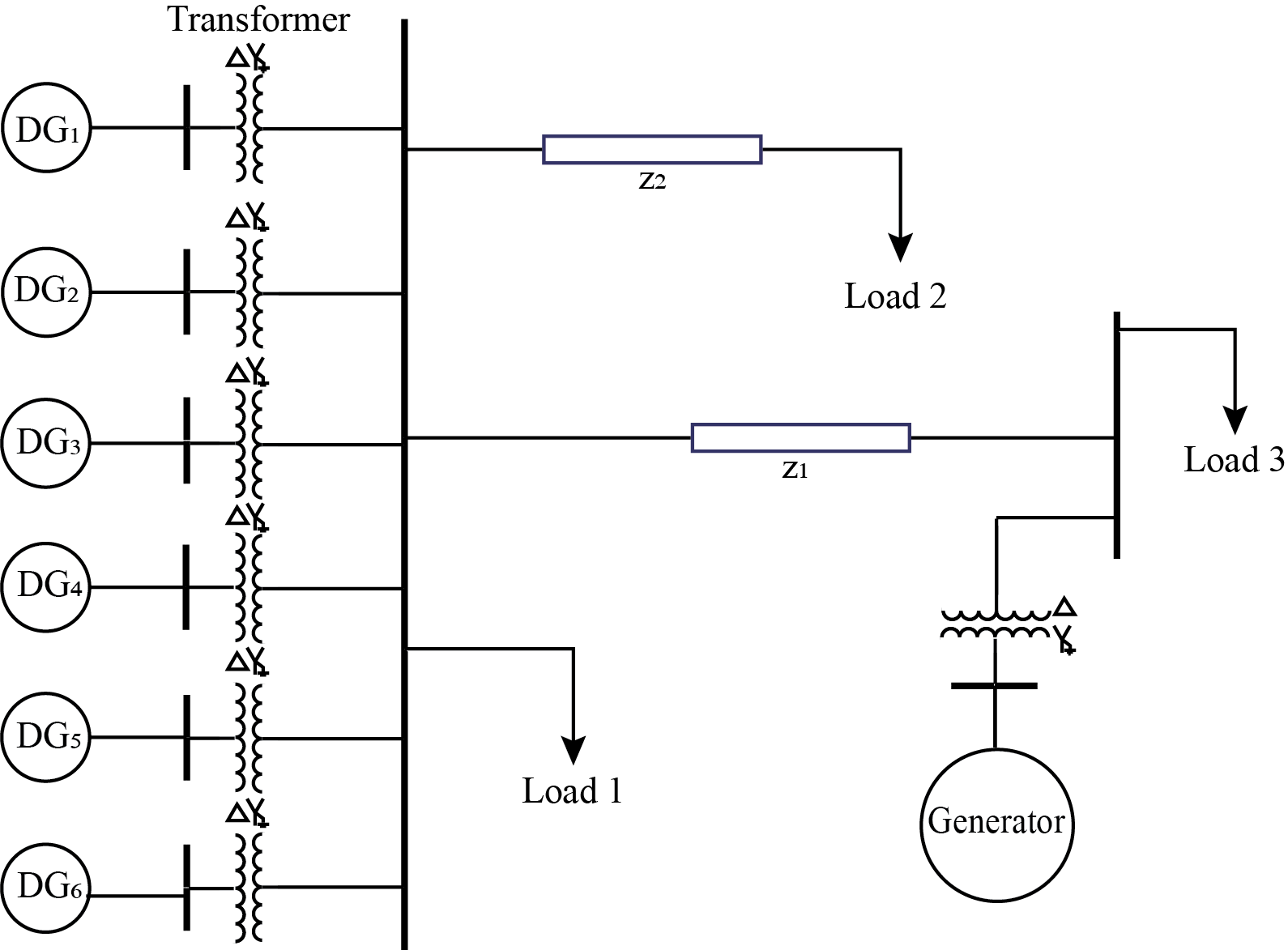}
	\caption{Simulated microgrid bus system}
	\label{fig: Bus system} 
\end{figure} 
We end this section with the following two observations:

\begin{remark}
The consensus algorithms proposed are independent from the load demand $P_L$ and its changes, referring to \ref{Sec: Static Consensus} and \ref{Sec: Dynamic Consensus}.
\end{remark}  

\begin{remark}
All agents in the cyber layer have access to the instantaneous value of $P_L$. The power generation command for strategies 1, 2 and 3 introduced in \ref{sec_rev2} and the control scheme in \ref{Sec: Dynamic Consensus} incorporate $P_L$ and $s_i(t)$ in the commanded power. Therefore, during consensus the power demand is met, irrespective of whether $P_L$ changes or not.
\end{remark}

\section{Controller Layout of Physical Layer} \label{Sec: section 4}
The proposed power control methods for DGs introduced in this study are required to be implemented on both cyber and physical layer of microgrids. The physical layer which includes DGs is where controllers are designed to control the output power of DGs. In this study, the problem of proportional power sharing is addressed in the grid-connected mode, hence the frequency and voltage of DGs are imposed by the main grid. Therefore, frequency and voltage control methods, such as droop control, are not considered in this study. Furthermore, the reactive power control in the grid-connected mode is not studied for the practical reason of availability of reactive power in the main grid. Therefore, the required reactive power of the microgrid can be maintained from the main grid.

The desired active power command of each DG$_i$, i.e $P_i^*$ for $i=1,2,\cdots,N$ is calculated by its corresponding agent, i.e $i^{th}$ agent in the cyber layer. Then, this signal of $P_i^*$ is sent to the power control block of DG$_i$ located in the physical layer. The power control block of DGs is represented in \cref{fig: Power control and abc to dq0}.
This block receives the voltage $V_{abc}$ and current $I_{abc}$ from the voltage and current measurement units installed on the output of each DG, as shown in \cref{fig: DG Power Circuit Diagram}. Figure \ref{fig: DG Power Circuit Diagram} also shows that each DG is connected to the main grid via a dedicated transformer to match the voltage between the DG and the main grid, as the output voltage of the main grid is significantly higher than the output voltage of DGs.
To control the generated power of a DG$_i$, i.e, $P_i$, it is required to control its output current since $V_{abc}$ and the frequency of the microgrid are fixed by the main grid. To achieve this, the desired active power command $P_i^*$ issued from $i^{th}$ agent is also considered as the other input in \cref{fig: Power control and abc to dq0}. Using the Phase-Locked-Loop (PLL) block, the signals in \cref{fig: Power control and abc to dq0} are converted to their equivalent values in the $dq0$ reference frame, i.e. $V_{dq}$ and $I_{dq}$.

Next, the outputs of the $V_{dq}$ and $I_{dq}$ are fed as inputs to \cref{fig: Current Control}. The parameters $C_1$, $C_2$ and the $PI$ controllers coefficients together with the upper and lower bounds of the saturation blocks in \cref{fig: Current Control} are all defined in \cref{Sec: Simulation}.
The outputs of the \cref{fig: Current Control}, regarded as the imaginary and real values of a complex number, are the inputs of the \cref{fig: Voltage Control}. These inputs are converted to the amplitude and phase angle of the same complex value. The amplitude and the phase signals, together with the voltage angle $\omega\,t$, obtained from the PLL in the \cref{fig: Power control and abc to dq0}, constitute the three phase signal fed to the PWM in \cref{fig: Voltage Control}. Finally, each PWM sends the switching signals to the three level inverter of its corresponding DG which is illustrated in \cref{fig: DG Power Circuit Diagram}. \vspace{-3mm}
\begin{figure} [t!]
	\begin{center} 
	\includegraphics[width=0.45\textwidth]{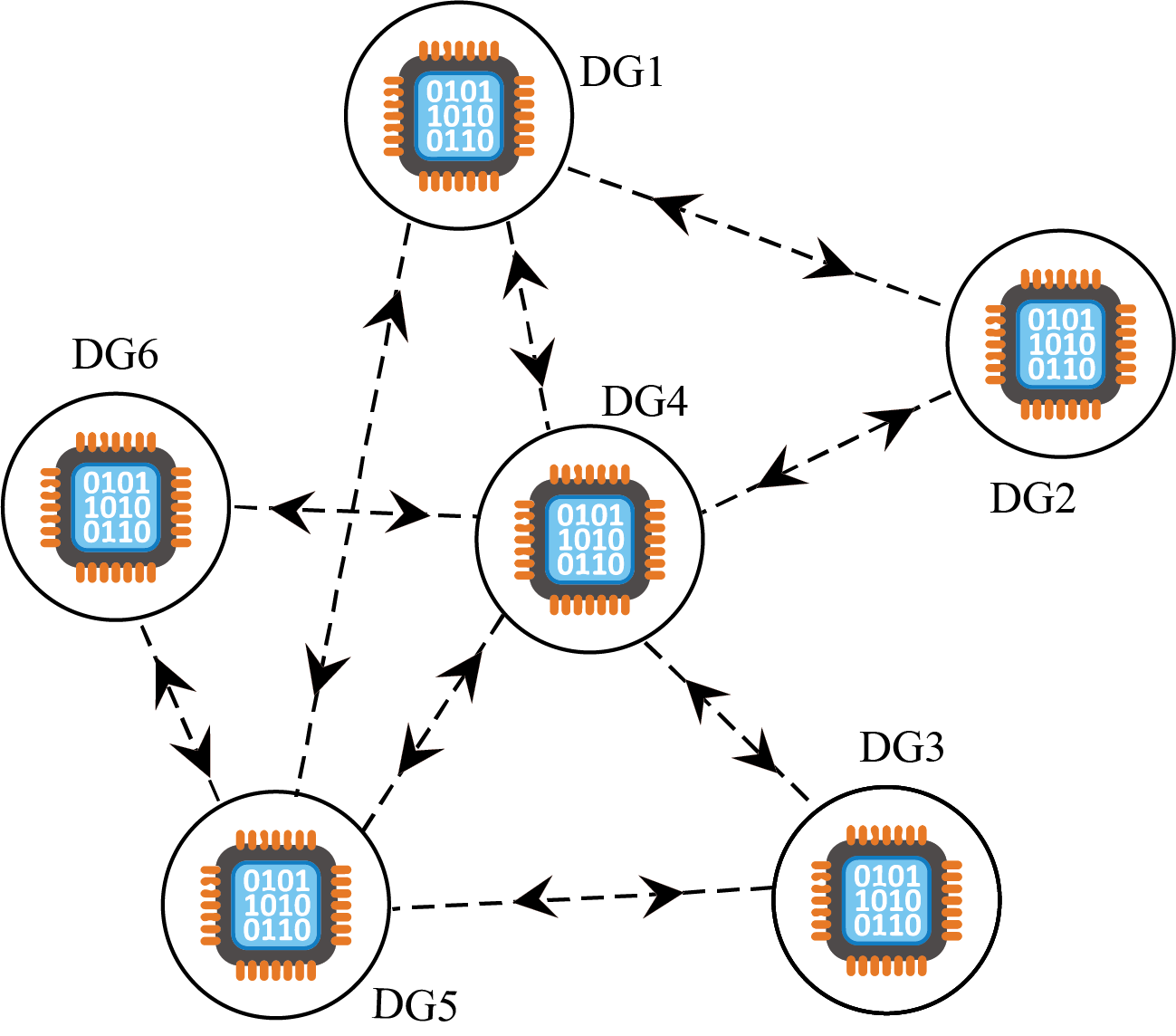}
	\end{center}
	\caption{Communication graph of the simulated $DG$s}
	\label{Fig: Topology of the Simulated microgrid} \vspace{-4mm}
\end{figure}
\begin{figure*}[!htp]  
	\centering
	\subfloat[]
	{
		   \includegraphics[width=.45\textwidth]{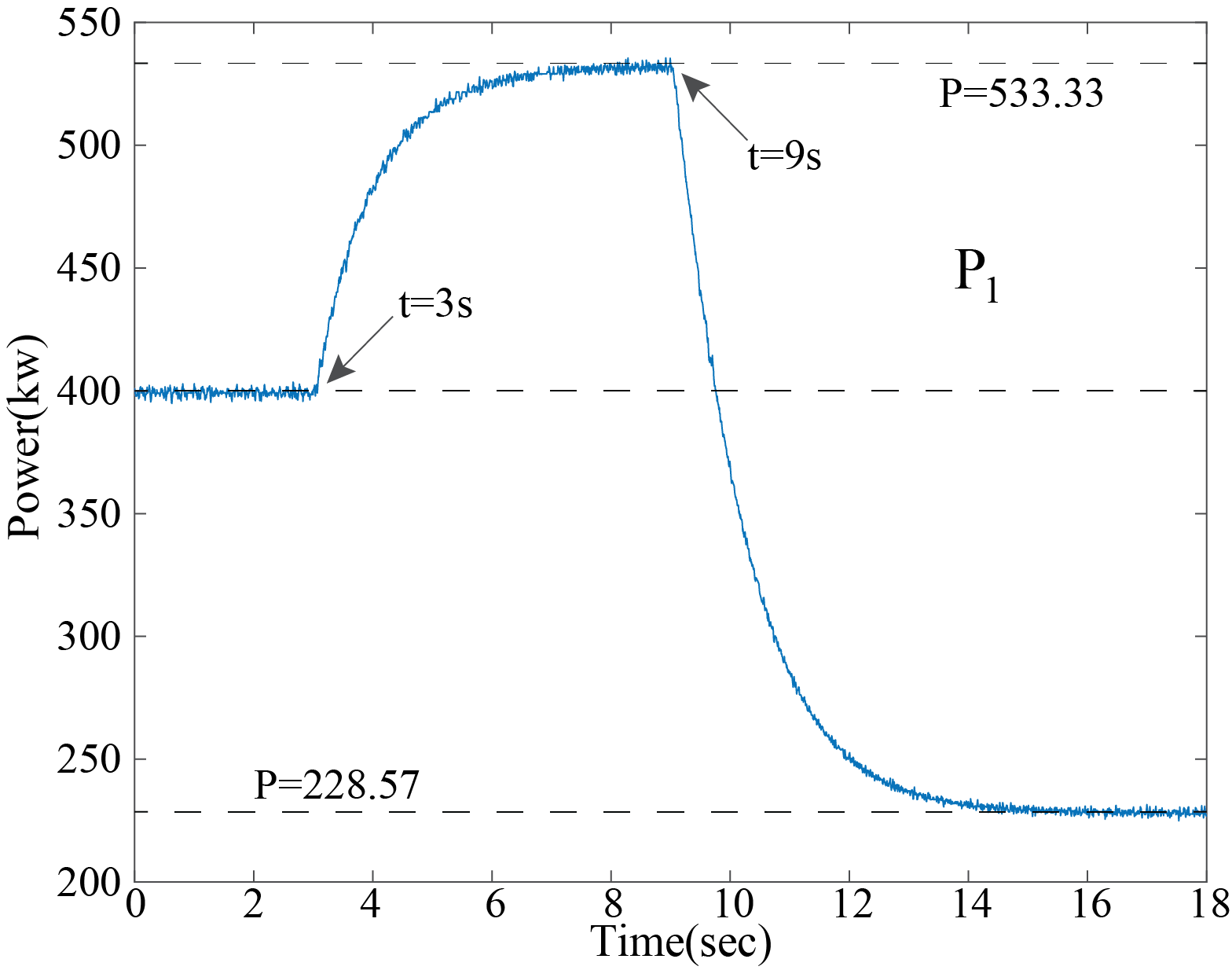}
		\label{fig: DG1 output power}
	}
	\hfill
	\subfloat[]
	{ 
   \includegraphics[width=.45 \textwidth]{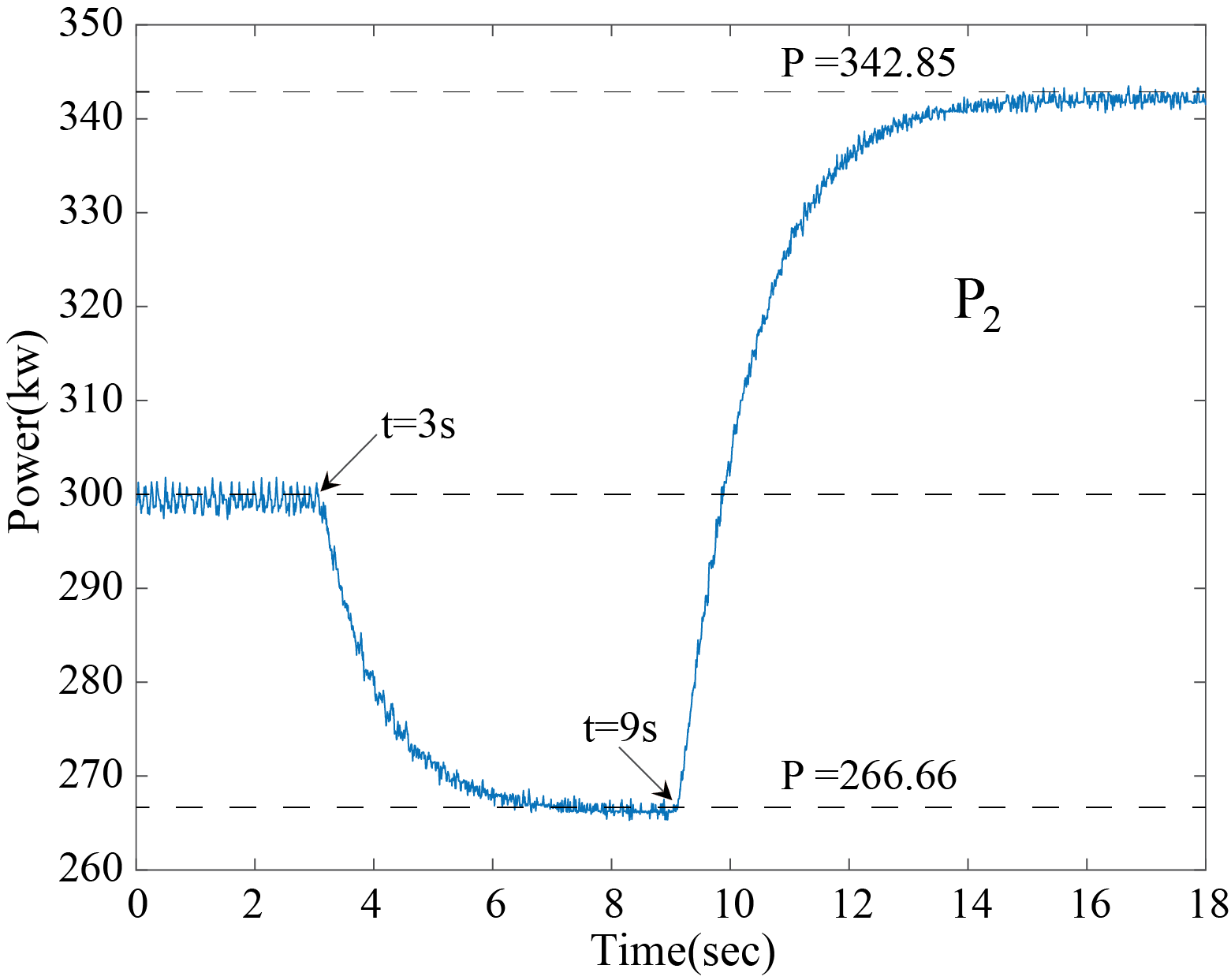}
		\label{fig: DG2 output power}
	}
	\hfill
	\subfloat[]
	{
   \includegraphics[width=.45 \textwidth]{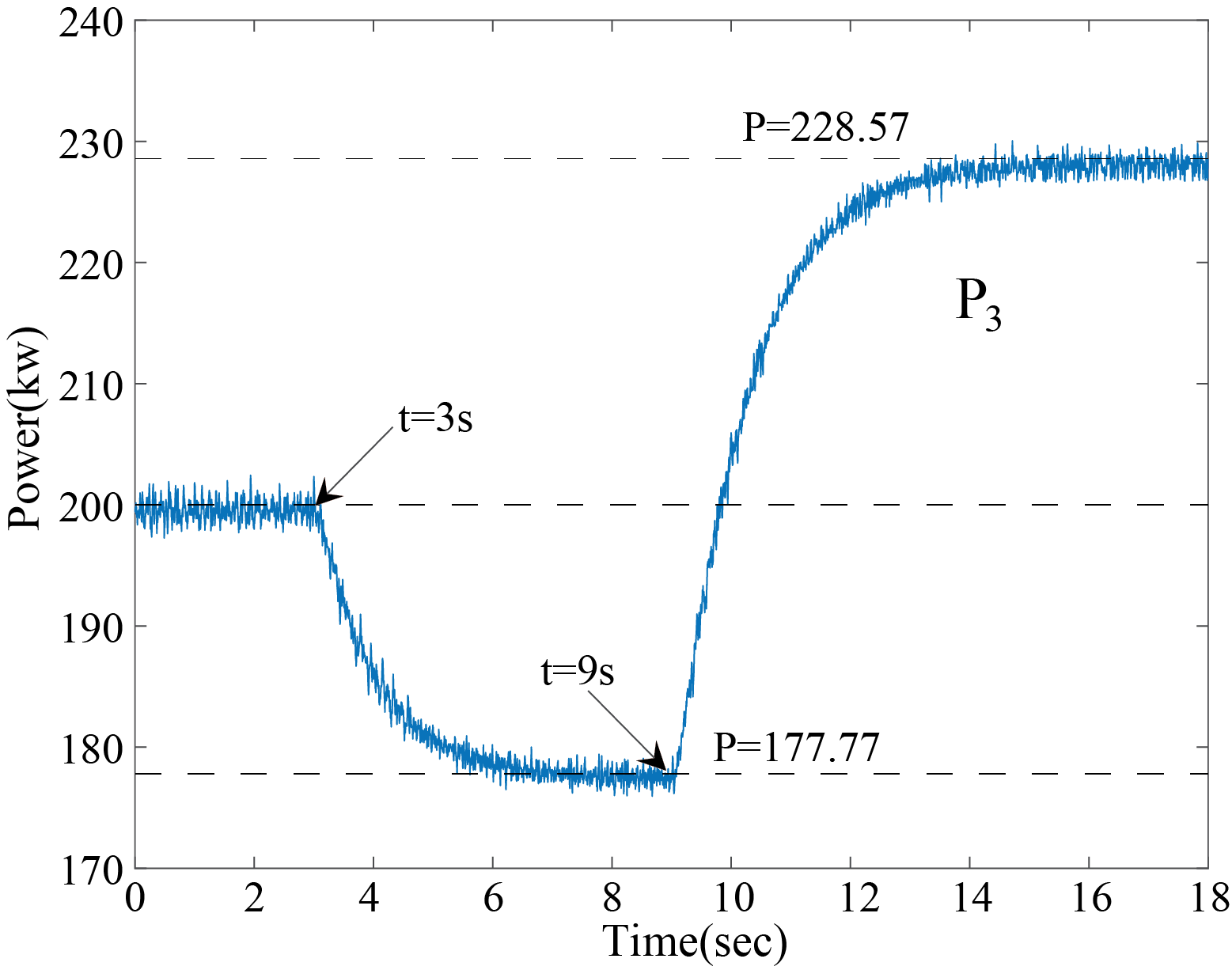}
		\label{fig: DG3 output power}
	} 
	\hfill
	\subfloat[]
	{
   \includegraphics[width=.45 \textwidth]{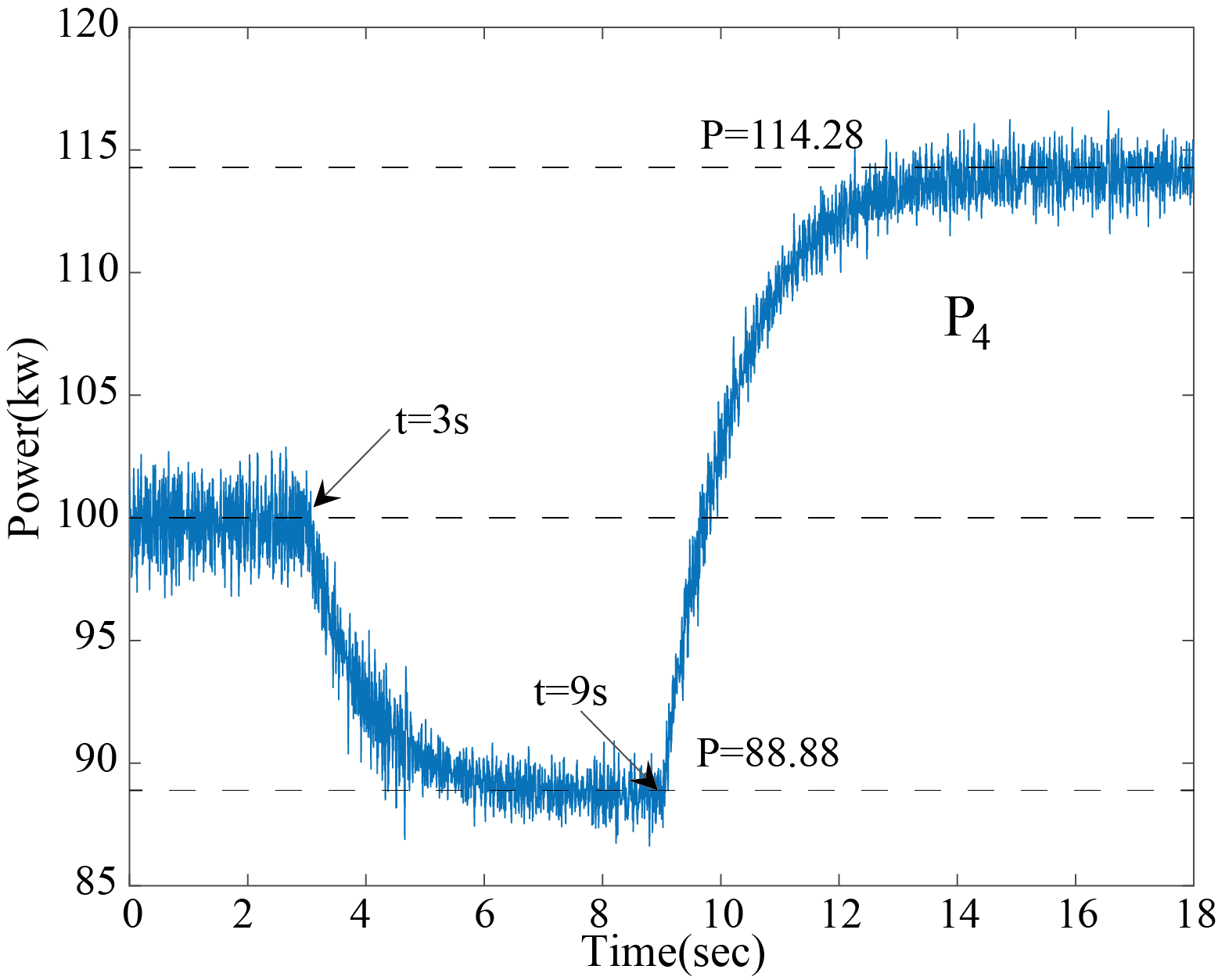}
		\label{fig: DG4 output power}
	} 
	\hfill
	\subfloat[]
	{
\includegraphics[width=.45 \textwidth]{DG4_Cnvgs_e1.png}
		\label{fig: DG5 output power}
	}   
	\hfill
	\subfloat[]
	{
		\includegraphics[width=.45 \textwidth]{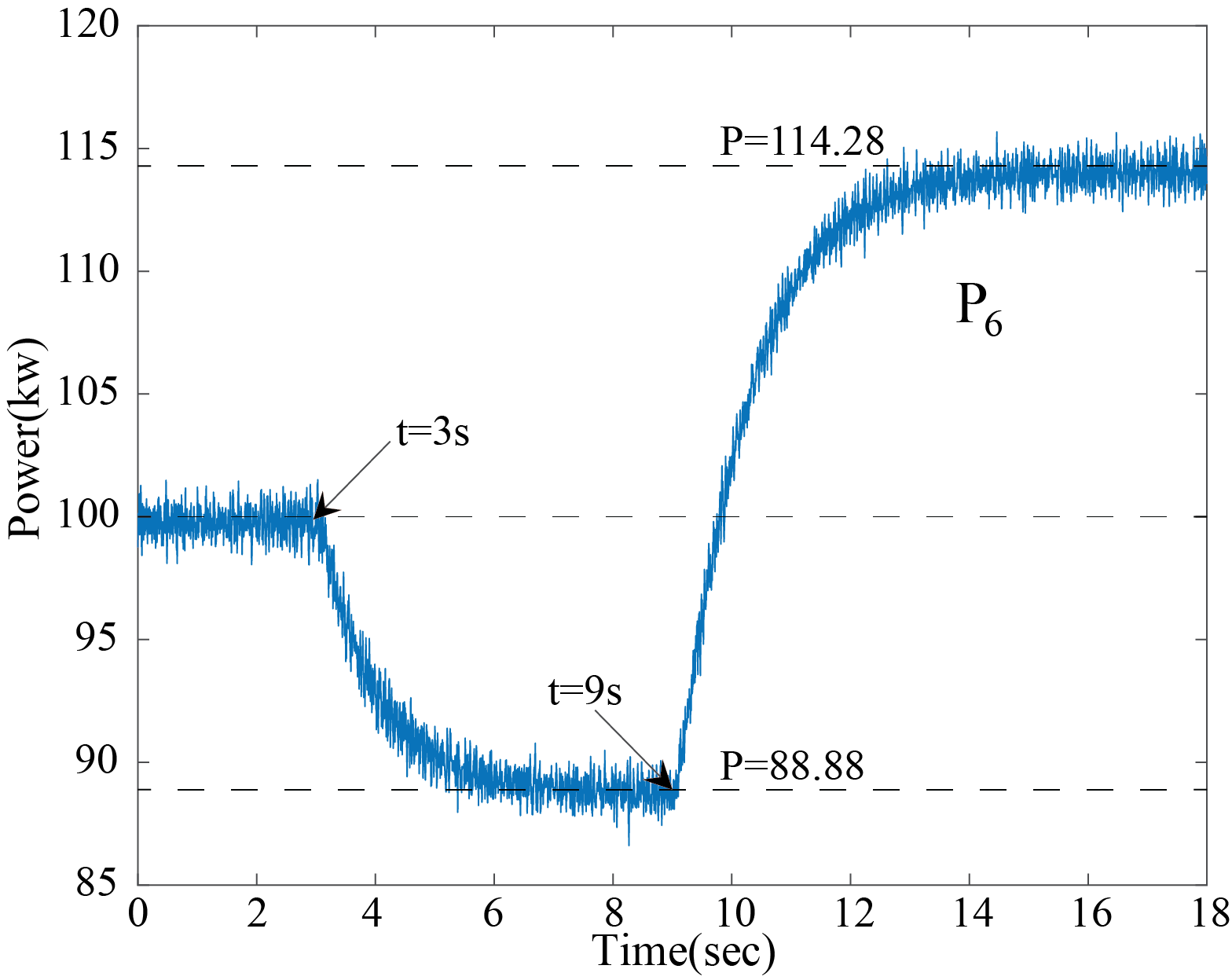}
		\label{fig: DG6 output power}
	}  
	\caption{Output powers $P_i$, $i=1, 2, \cdots, 6$, under a variation in $P_{1,max}$ are    
		depicted in the figures 
		\protect\subref*{fig: DG1 output power}, 
		\protect\subref*{fig: DG2 output power}, 
		\protect\subref*{fig: DG3 output power}, 
		\protect\subref*{fig: DG4 output power}, 
		\protect\subref*{fig: DG5 output power} and 
		\protect\subref*{fig: DG6 output power}, respectively.   
	}
	\label{fig:oPowerDsdisplay}
\end{figure*}
\begin{figure}[htp]
	\centering
   \includegraphics[width=.7 \textwidth]{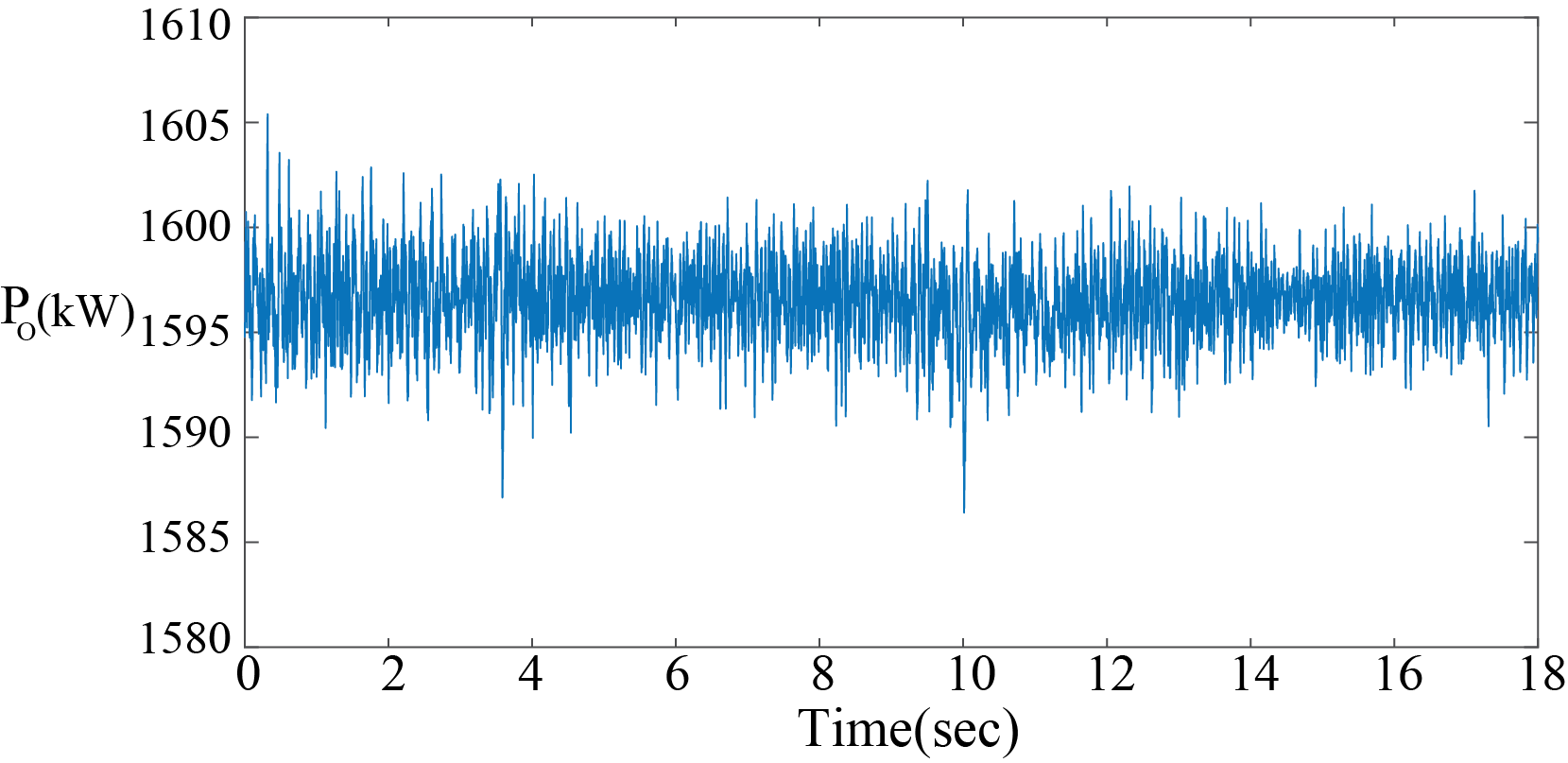}
	\caption{Microgrid total output power $P_O$ obtained from the proposed control method of \cref{Sec: Dynamic Consensus} }
	\label{Fig: Total Output Power} 
\end{figure} 
\begin{figure*}[!htp]
	\centering
	\subfloat[]
	{
   \includegraphics[width=.45\textwidth]{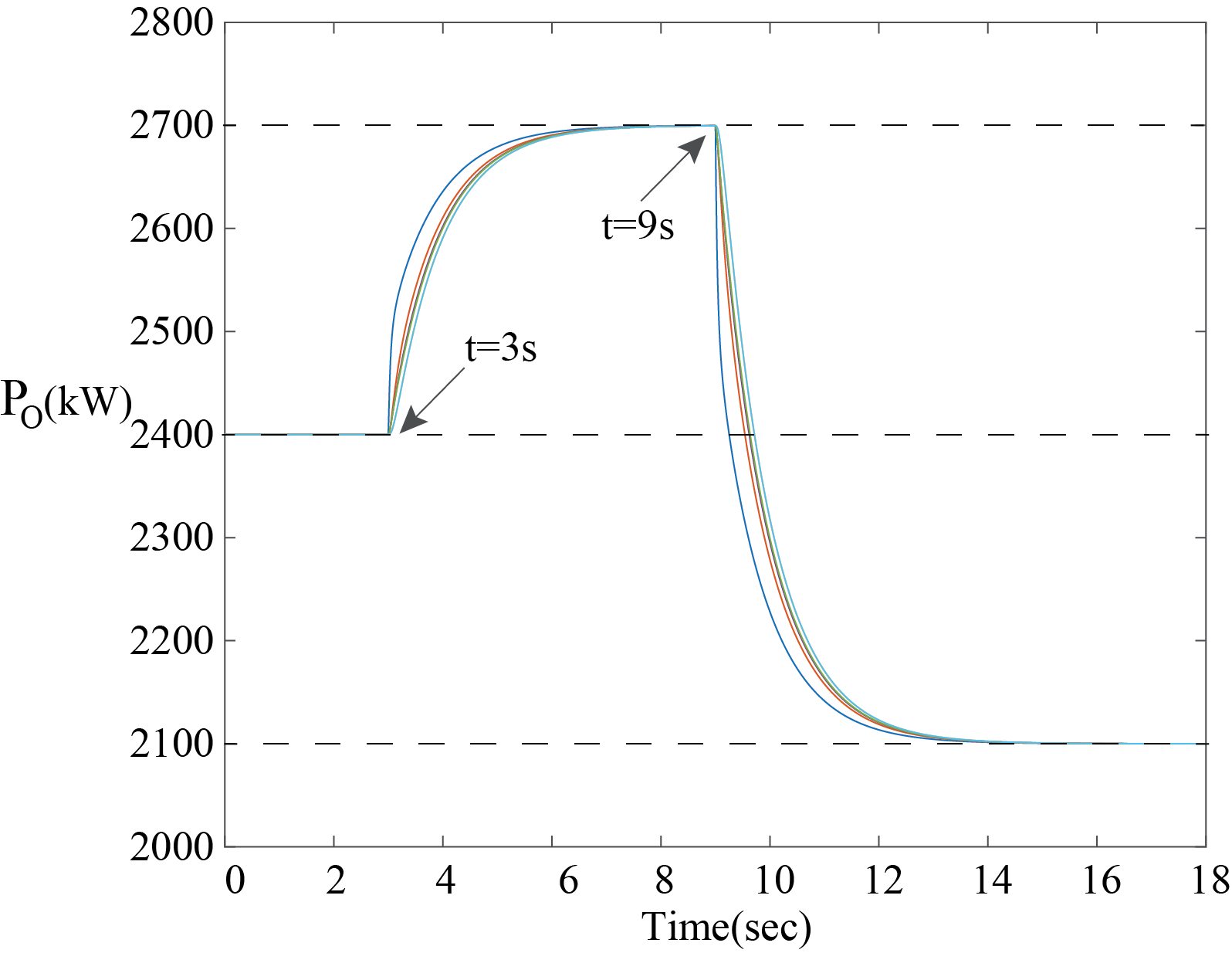}
		\label{fig: Consensus Convergence}	
	}
	\hfill
	\subfloat[]   
	{ 
     \includegraphics[width=.45 \textwidth]{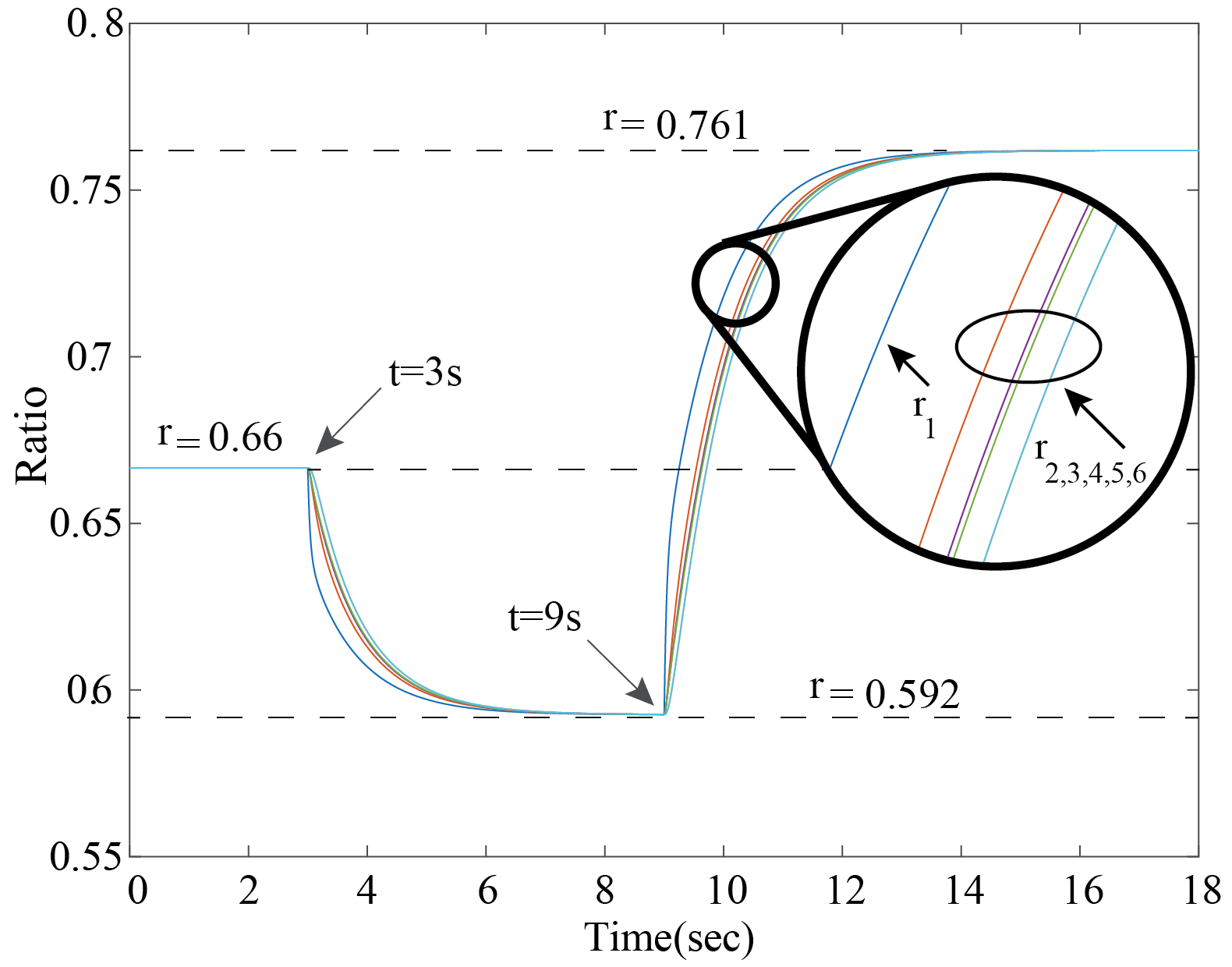}
		\label{fig: Zoomed in Consensus Convergence}
	}
	\label{Consensus}
	\caption{Consensus trajectories of agents on $\tilde{P}_T$ from (\ref{Eq: update of Si in realtime}) and trajectories of $r_i$, $i=1,2,\cdots,6$, from (\ref{Eq: Keeping Power Sharing Control}). 
		(\protect\subref*{fig: Consensus Convergence}) \textnormal{Consensus trajectories on $\tilde{P}_T$ } 
		(\protect\subref*{fig: Zoomed in Consensus Convergence}) \textnormal{The signals $r_i$ } 
	}
\end{figure*}
\begin{figure}[t]
	\centering
	\includegraphics[width=.45 \textwidth]{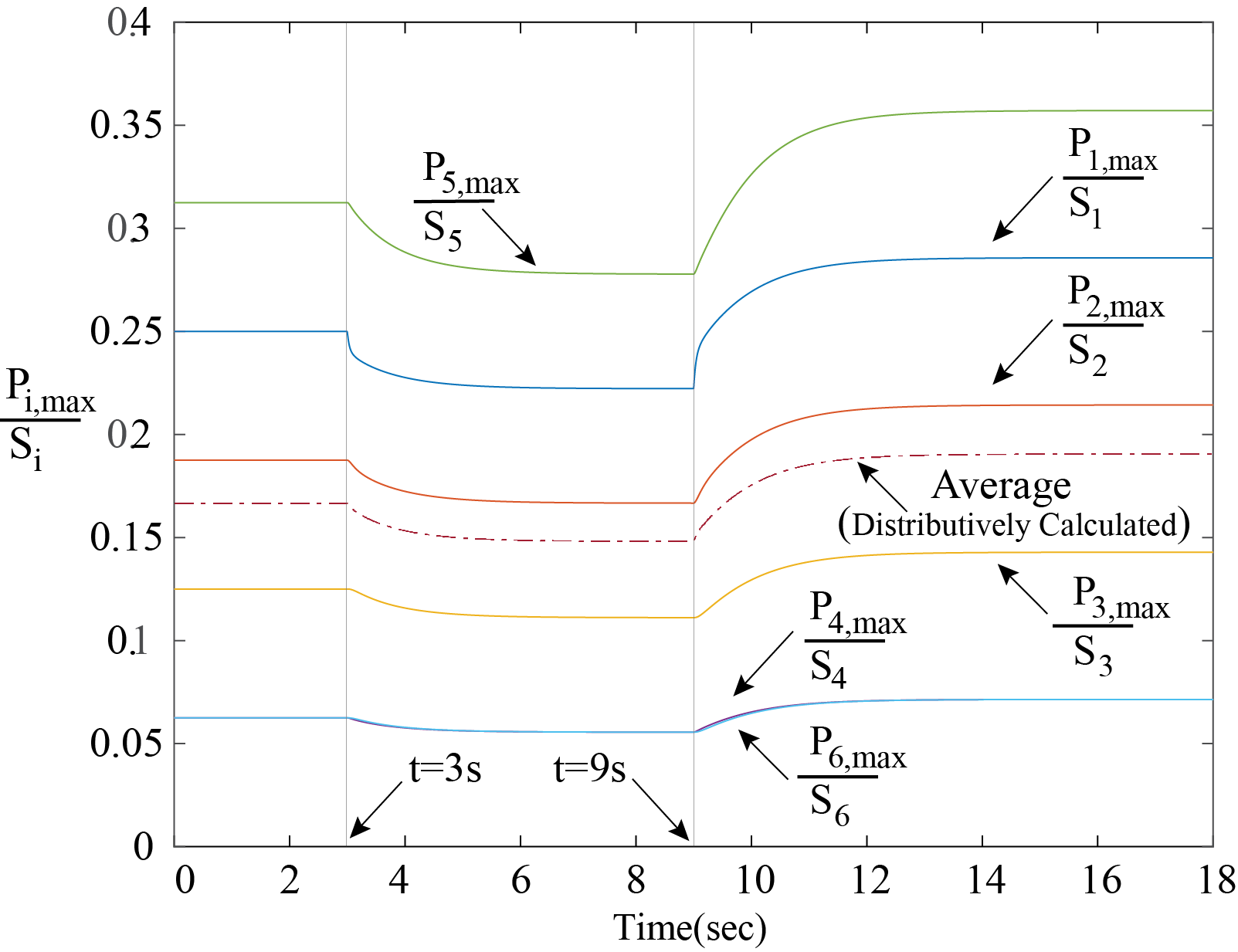}
	\caption{Individual ${P_{i,max}}/{s_i(t)}$  and the average ${P_{i,max}}/{s_i(t)}$ at each time step, as defined in (\ref{Eq:  Accumulation of Terms except pk/sk}) and (\ref{Eq: Average})}
	\label{fig: pisi ratios average} 
\end{figure}
\section{Simulations} \label{Sec: Simulation}
In this section, the performance of the proposed control methods explained in \cref{Sec: Dynamic Consensus} is evaluated through the simulation of a microgrid consisting of six inverter-based DGs shown in \cref{fig: Bus system}. 
The model layout in \cref{fig: Bus system} is inspired from Matlab-based example available in \cite{Matlab_Sims} 
and the studies in \cite{schiffer2015voltage,Nasirian2016Droop}.
Then, the performances of the strategies 1 and 3, provided in (\ref{Eq: strategy 3}) and (\ref{Eq: Strategy 1}) respectively, are juxtaposed with the performance of the controller in \cref{Sec: Dynamic Consensus}. The simulations are accomplished using the Simscape toolbox of Matlab. The simulated DGs are numbered from 1 to 6 and are connected to the main grid in parallel as depicted in \cref{fig: Bus system}. Each DG has a corresponding agent in the cyber layer where the updated value of the desired output power is computed by the agents using the information obtained through their bidirectional communication structure, as shown in \cref{Fig: Topology of the Simulated microgrid}. Note that the communication graph of the DGs in \cref{Fig: Topology of the Simulated microgrid} is connected per its definition in \cref{Sec: Preliminary Definitions}.

As the communication graph in \cref{Fig: Topology of the Simulated microgrid} is a bidirectional graph, per \cref{Sec: Preliminary Definitions}, the adjacency matrix of the graph is symmetric. The adjacency and degree matrices are chosen as

\small
\begin{equation} \label{Eq: Adjacency and Diagonal}
A =\!\!\left[ \begin{array} {ccccccc} \!\!\!\! 0 & \!\!\! 6 & \!\! \! 0 & \!\!\! 6 & \!\!\! 6 & \!\!\! 0 \!\! \!\!\\ \!\!\!\! 6 & \!\!\! 0 & \!\!\! 0 & \!\! \! 6 & \!\!\! 0 & \! \!\! 0 \!\! \!\!\\ \!\!\!\! 0 & \!\!\! 0 & \!\!\! 0 & \!\!\! 6 & \!\!\! 6 & \!\!\! 0 \!\! \!\!\\ \!\!\!\! 6 & \!\!\! 6 & \!\!\! 6 & \!\!\! 0 & \!\!\! 6 & \!\!\! 6 \!\!\!\! \\\!\! \!\! 6 & \!\!\! 0 & \!\! \!6 & \!\!\! 6 & \!\!\! 0 & \!\!\! 6 \!\! \!\!\\ \!\!\!\! 0 & \!\!\! 0 & \!\!\! 0 & \!\!\! 6 & \!\!\! 6 & \!\!\! 0 \!\! \end{array}\right]\!\!\!,\,\, D=\!\! \left[ \begin{array} {ccccccc} \!\!\!\! 18 & \!\!\!\!\! 0 & \!\!\!\!\! 0 & \!\!\!\!\! 0 & \!\!\!\!\! 0 & \!\!\!\!\! 0 \!\!\!\! \\ \!\!\!\! 0 & \!\!\!\!\! 12 & \!\!\!\!\! 0 & \!\!\!\!\! 0 & \!\!\!\!\! 0 & \!\!\!\!\!0 \!\!\!\!\\ \!\!\!\! 0 & \!\!\!\!\! 0 & \!\!\!\!\!12 & \!\!\!\!\! 0 & \!\!\!\!\! 0 & \!\!\!\!\!0 \!\!\!\! \\ \!\!\!\! 0 & \!\!\!\!\! 0 & \!\!\!\!\! 0 & \!\!\!\!\! 30 & \!\!\!\!\! 0 & \!\!\!\!\! 0 \!\!\!\! \\ \!\!\!\! 0 & \!\!\!\!\! 0 & \!\!\!\!\! 0& \!\!\!\!\! 0& \!\!\!\!\! 24 & \!\!\!\!\! 0 \!\!\!\! \\ \!\!\!\! 0 & \!\!\!\!\! 0 & \!\!\!\!\! 0& \!\!\!\!\! 0& \!\!\!\!\! 0 & \!\!\!\!\! 12 \!\!\!\! \end{array}\right] \!\!
\end{equation}
\normalsize
From  \cref{Sec: Preliminary Definitions}, the correlated Laplacian matrix $L=A-D$ is defined as 
\begin{equation} \label{Eq:Laplacian Matrix}
L=\left[ \begin{array} {ccccccc} 18 &-6& 0 &-6& -6& 0 \\ -6& 12& 0 &-6 &0& 0\\ 0& 0& 12& -6& -6& 0  \\-6& -6& -6& 30& -6 &-6 \\-6& 0& -6& -6& 24& -6\\ 0& 0& 0& -6& -6& 12 \end{array}\right]
\end{equation}
For the simulations, we set $h=10$ in (\ref{Eq: Matrix Format of the Agents Communication Dynamics}), (\ref{Eq: update of Si in realtime}) and (\ref{Eq: discretized update of consenus}). Let the maximum capacity of the DGs $P_{i,max}$ for $i=1,2\cdots, 6$ be    
\begin{equation}
\begin{split}
P_{1,max}=600 \,kw  \;\; P_{2,max}=450 \, kw \;\; P_{3,max}=300 \, kw  \\
P_{4,max}=150 \, kw \;\; P_{5,max}=750 \, kw \;\; P_{6,max}=150 \, kw
\end{split}
\end{equation}
Thus, the maximum capacity of the whole microgrid is $P_T=\sum_{i=1}^{6} P_{i,max}=2400 \, kw$, and assume the load demand is $P_{L}=1600 \,kw$.
According to \cref{Sec: Power Defenition}, the agents have knowledge of $P_L$ at all times and $P_T$ at initial time. Therefore, each agent is able to compute the proportional power share ratio $r=\frac{P_L}{P_T}$ defined in (\ref{Eq: Proportional Ratio Root}), independently, which is $\frac{2}{3}$, initially. Therefore, the output power of each DG$_i$ for $i=1,2,\cdots,6$, based on the proportional power sharing, must be,
\begin{equation} \label{eq: initial power dispatch}
\begin{aligned}
P_{1}=400\,kW \quad P_{2}=300\,kW \quad P_{3}=200\,kW  \\
P_{4}=100\,kW \quad P_{5}=500\,kW \quad P_{6}=100\,kW 
\end{aligned}
\end{equation}
During $t=[0,3]sec$, the power capacity of each  DG remains unchanged and hence, each of the DGs generates its active power share as calculated in (\ref{eq: initial power dispatch}). At $t= 3 sec$ the power capacity of DG$_1$ undergoes a step change, thereby $P_T$ has an increment of $300\,kW$, and at $t=9\,sec$
we introduce a decrement of $600 kW$ to the capacity of DG$_1$.

The simulated microgrid consists of several components such as inverters, output filters of inverters, transformers, PWM, PI controllers, line impedance, loads, DC resources, measurement units, PLL and abc/dq0 converters. To emulate the main grid a dispatchable generator is considered in the simulation as shown in \cref{fig: Bus system}. The parameters of the transformer that connects the main grid to the distribution system and those of the transformers which connect the DGs to the distribution system are given in \cref{Tab: Parameters of the microgrid}. The parameters of the distribution system are provided in \cref{Tab: ParametersoftheGrid}. The PI controllers depicted in \cref{fig: Current Control} are identical. The PI controllers of all DGs are also identical, meaning they all have the same $P$ and $I$ gains, chosen as $k_P=0.3$ and $k_I=30$, respectively. 

The energy resource of each DG is simulated as a DC power source, then by utilizing an inverter, the DC current converts to the AC current, as shown in \cref{fig: DG Power Circuit Diagram}. In the same figure, to remove the harmonics from the output power of the inverter, an output filter is applied and then connected to the main grid via a transformer, as illustrated in \cref{fig: DG Power Circuit Diagram}. The output filter consists of RL and  RC branches. The resistive and inductive elements of each RL component are set as R$_1=5.4946\times 10^{-4} \Omega$ and L$=1.4575\times 10^{-4} H$, respectively. The RC components for the output filters of each individual DG$_i$ arranged in the delta format have $P_i (W)$ and reactive power $Q_i(kVar)$ as given in \cref{tab: Active and Reactive powers of filter}. In \cref{fig: Current Control}, $C_1=0.0039$, $C_2 = 0.21$, and the upper and lower limit of the saturation blocks are $+1.5$ and  $-1.5$, respectively. The power control of each DG established in the physical layer is depicted in \cref{fig: Whole Power Control Diagram}.
\begin{table} 
	\centering
	\caption{Parameters of the Transformers}
	\begin{tabular}[t]{c|c|c}
		\multicolumn{3}{c}{Transformer Connected to DGs}\\
		\hline
		\multicolumn{2}{c|} {\makecell {Parameter}} & Value \\
		\hline
		\multicolumn{2}{c|} {\makecell {Nominal Power (kVA)}}& 100 \\
		\hline
		\multicolumn{2}{c|} {\makecell{Nominal Frequency (Hz)}} & 60\\
		\hline
		\multirow{3}{4em}{\makecell{winding 1}}& $V_{1,rms, ph-ph, kV}$ & 25\\
		&$R_{1} (pu)$ & 0.0012\\
		&$L_{1} (pu)$ & 0.03\\
		\hline
		\multirow{3}{4em}{\makecell{winding 2}}& $V_{1,rms, ph-ph, V}$ & 270\\
		&$R_{2} (pu)$ & 0.0012\\
		&$L_{2} (pu)$ & 0.03\\
		\hline
		\multicolumn{2}{c|} {\makecell {Magnetization resistance  Rm (pu)}}  & 200\\
		\hline
		\multicolumn{2}{c|} {\makecell {Magnetization inductance  Lm (pu)}}  & 200 \vspace{5mm}
	\end{tabular}
	\begin{tabular}[t]{c|c|c}
		\multicolumn{3}{c}{Transformer Connected to Dispatchable Generator}\\
		\hline
		\multicolumn{2}{c|} {\makecell {Prameter}} & Value \\
		\hline
		\multicolumn{2}{c|} {\makecell {Nominal Power (kVA)}}& 47000 \\
		\hline
		\multicolumn{2}{c|} {\makecell{Nominal Frequency (Hz)}} & 60\\
		\hline
		\multirow{3}{4em}{\makecell{winding 1}}& $V_{1,rms, ph-ph, kV}$ & 1200\\
		&$R_{1} (pu)$ & 0.0026\\
		&$L_{1} (pu)$ & 0.08\\
		\hline
		\multirow{3}{4em}{\makecell{winding 2}}& $V_{2,rms, ph-ph, kV}$ & 25\\
		&$R_{2} (pu)$ & 0.0026\\
		&$L_{2} (pu)$ & 0.08\\
		\hline
		\multicolumn{2}{c|} {\makecell {Magnetization Resistance  Rm (pu) }} & 500\\
		\hline
		\multicolumn{2}{c|} {\makecell {Magnetization Inductance  Lm (pu)}} & 500\\
	\end{tabular}
	\label{Tab: Parameters of the microgrid}
\end{table}
\begin{table}
	\caption{Parameters of the Grid}
	\centering 
	\begin{tabular}[htbp]{c|c|c} 
		\hline
		\multicolumn{2}{c|} {\makecell{Parameter}} & Value \\
		\hline
		\multicolumn{2}{c|}{\makecell{Load 1 Nominal Voltage ($kV_{ph-ph}$)}}& 25 \\
		\hline
		\multicolumn{2}{c|}{\makecell{Load 1 Active Power P (kW)}} & 250\\
		\hline
		\multicolumn{2}{c|}{\makecell{Load 2 Active Power P (kW)}} & 2000\\
		\hline
		\multicolumn{2}{c|}{\makecell{Load 3 Power S (kVA)}} & 30000+j2000\\
		\hline
		\multirow{4}{6em}{\makecell{ Line 1 $Z_1$ \\ Positive and\\ Zero Sequence}} & Length (km)& 8 \\
		&$R (\Omega/km)$& [0.1153 0.413] \\
		&$L (H/km)$ &[1.05e-3 3.32e-3]\\
		&$C (F/km)$ &[11.33e-009 5.01e-009]\\
		\hline
		\multirow{4}{6em}{\makecell{Line 2 $Z_2$\\ Positive and \\ Zero Sequence}} & Length (km)& 14 \\
		&$R (\Omega/km)$& [0.1153 0.413] \\
		&$L (H/km)$ &[1.05e-3 3.32e-3]\\
		&$C (F/km)$ &[11.33e-009 5.01e-009]\\
		\hline
		\multicolumn{2}{c|}{\makecell{Nominal Frequency (Hz)}} & 60
	\end{tabular}
	\label{Tab: ParametersoftheGrid}
\end{table}

\begin{table} 
	\centering 
	\caption{Active and reactive powers of RC components of each output filter}
	\begin{tabular}[htbp]{c|c|c}
		\hline
		DG Number&Active Power(W)&Reactive Power(kVar)\\
		\hline
		1& 400 &20\\
		\hline
		2&200&10\\
		\hline
		3&600&30\\
		\hline
		4&500&25\\
		\hline
		5&300&15\\
		\hline
		6&400&20\\
	\end{tabular}
	\label{tab: Active and Reactive powers of filter} \vspace{-3mm}
\end{table}
\begin{figure}[htbp]
	\centering
	\subfloat[]
	{
		\includegraphics[width=.45 \textwidth]{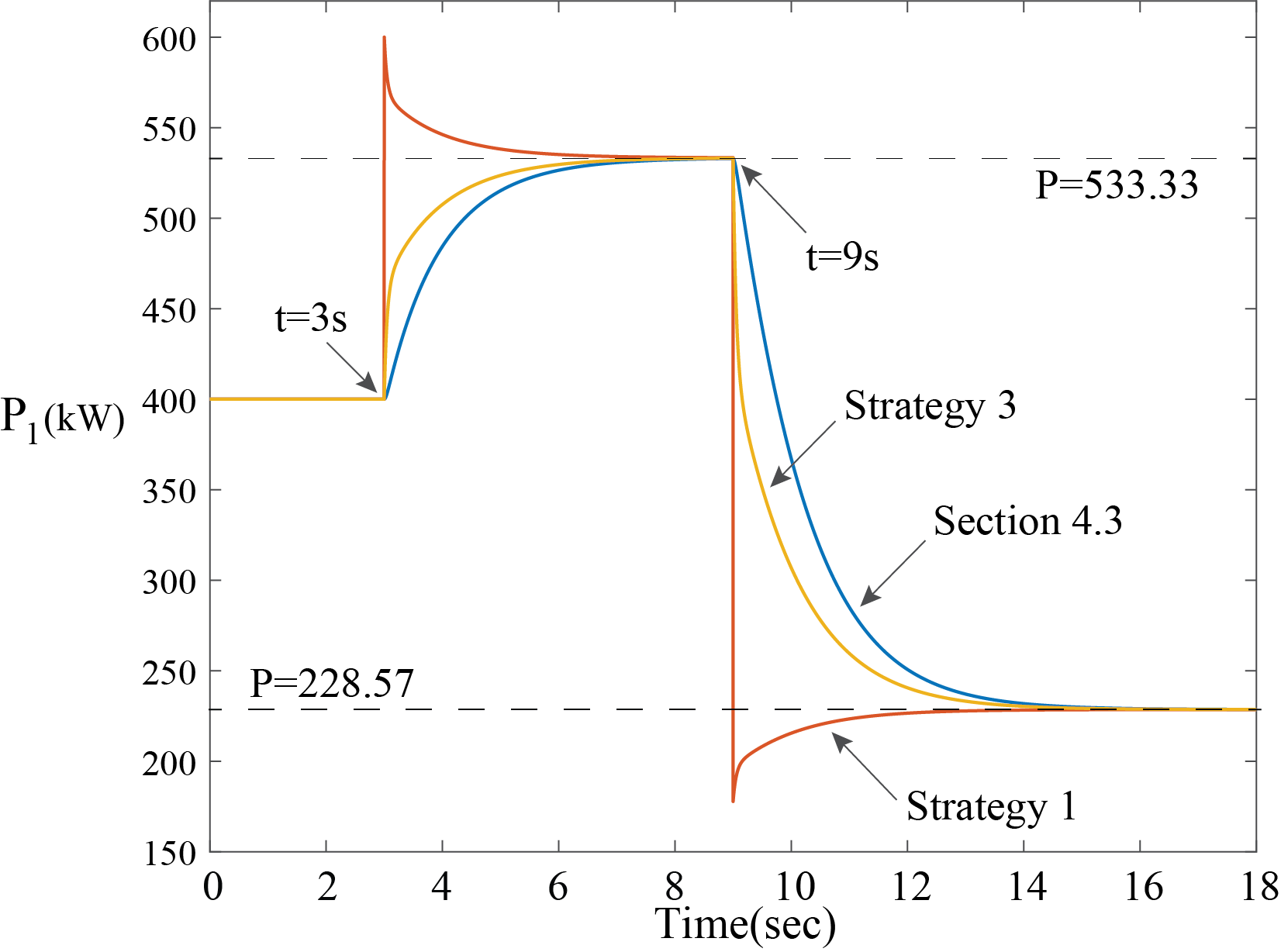} 
		\label{fig: DG1 Different Methods}
	}
	\hfill
	\subfloat[]
	{ 
		\includegraphics[width=.45 \textwidth]{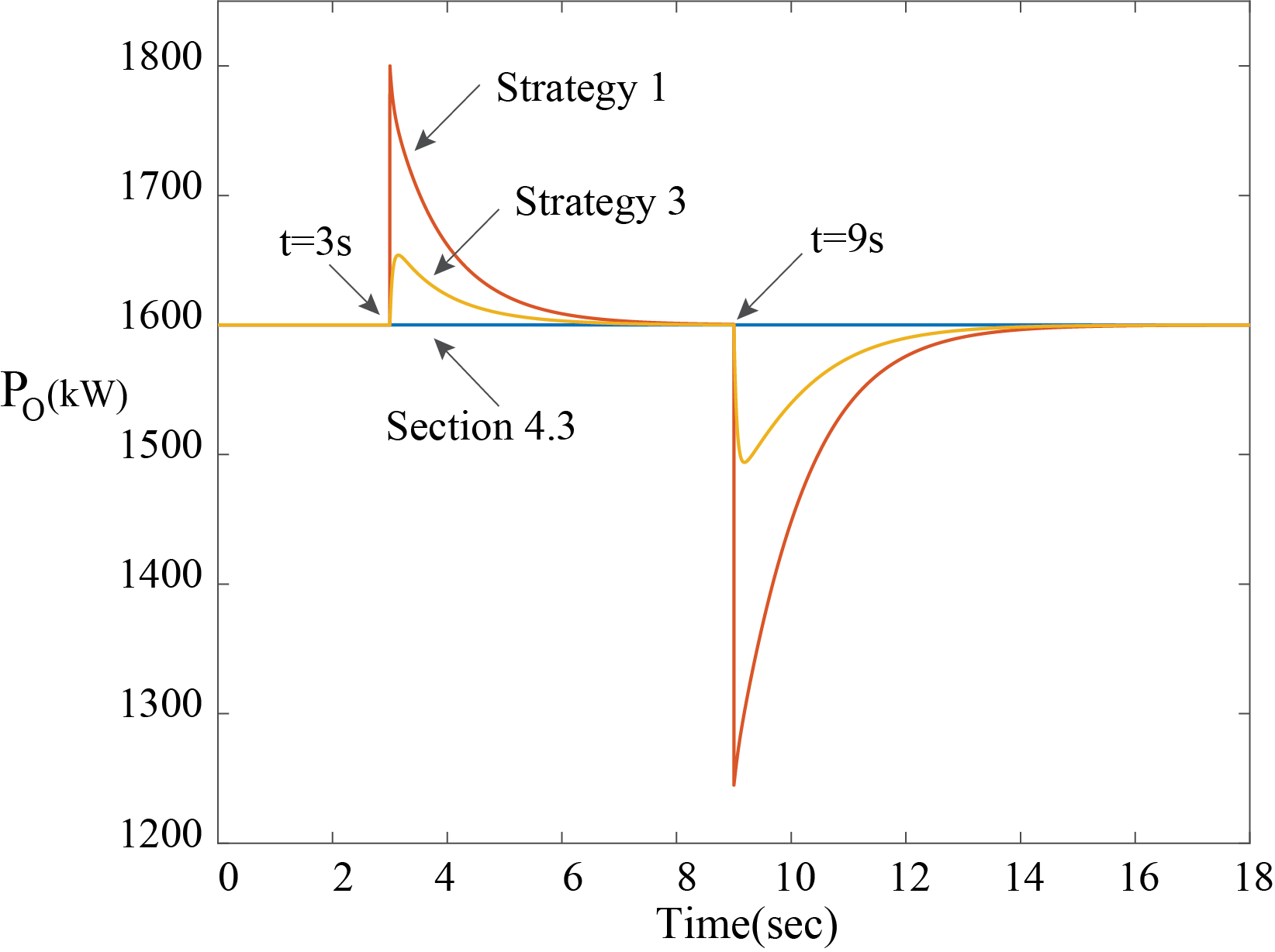}
		\label{fig: Output Sum Diff Method}
	}
	\label{fig: Different Method}    
	\caption{ (\protect\subref*{fig: DG1 Different Methods})  \textnormal{Output power $P_1$ according to Strategies 1, 3 and the proposed controller of \cref{Sec: Dynamic Consensus}} and (\protect\subref*{fig: Output Sum Diff Method}) \textnormal{Their corresponding microgrid total power $P_O$} 
	} \vspace{-6mm}
\end{figure}
Starting from $t = 3\, sec$, $P_{1,max}$ increases from $600\,kW$ to $900\,kW$. Therefore, the microgrid maximum power capacity increases from $2400\,kW$ to $2700\,kW$. 
Based on the approach explained in \cref{Sec: Dynamic Consensus}, the finite-time algorithm in (\ref{finite})-(\ref{Eq: final average finite step})
is embedded in the consensus algorithm (\ref{Eq: update of Si in realtime}) to have DGs apply the proposed control law in (\ref{Eq: Keeping Power Sharing Control}) and (\ref{eq: proposed algorithm}), distributively. The results of the simulations are shown in \cref{fig:oPowerDsdisplay} where $P_1$ increases and $P_i$, $i=2, 3, 4, 5, 6$, decreases. During $t=[3, 9]\,sec$, the microgrid output power $P_O$ remains almost equal to $P_L=1600\,kW$, as shown in \cref{Fig: Total Output Power}. The slight difference between $P_O$ and $P_L$ is due to resistive losses in the DGs due to the resistor elements shown in \cref{fig: DG Power Circuit Diagram}. Considering (\ref{Eq: update of Si in realtime}), the six estimation variables $s_i$ for all $i=1,2,\cdots,6$ are updated through the information exchange until they reach a consensus concerning $\tilde{P}_T=2700 kW$ as demonstrated in \cref{fig: Consensus Convergence}. The ratios of $r_i=\frac{P_{L}}{s_{i,max}}$, for $i = 1, 2, \cdots, 6$, are shown in \cref{fig: Zoomed in Consensus Convergence}, where before reaching a consensus during $t = (3, 8)\,sec$, these ratios are different, however they become almost identical during $t=[8,9]\,sec$.
Figure \ref{fig:oPowerDsdisplay} illustrates that after the consensus algorithm converges, the steady state values of the output powers of DGs are $P_{1}=533.33\,kW$, $P_{2}=266.66\,kW$, $P_{3}= 177.77\,kW$, $P_{4}=88.88\,kW$, $P_{5}=444.44\,kW$ and $P_{6}=88.88\,kW$. Recalling the finite-time average consensus algorithm is embedded in the consensus algorithm, at each time step of evolution of consensus algorithm, the finite-time algorithm is applied. Through this approach, the agents compute the average of $\frac{P_{i,max}}{s_i(t)}$ for $i=1, 2, \cdots, 6$ in a distributed way where its corresponding result is illustrated in \cref{fig: pisi ratios average}.

The next variation of the microgrid maximum power capacity occurs at $t=9s$ where $P_{1,max}$ decreases for $-600 kW$. Therefore, starting from $t=9\,sec$, the current capacity of the microgrid which is $P_T=2700\,kW$ changes to $\tilde{P}_T=2100 kW$. Then, similar to the same procedure adopted in reaction to a change in a microgrid capacity, the control method in (\ref{Eq: Keeping Power Sharing Control}) and (\ref{eq: proposed algorithm}) is triggered. Hence, during $t=[9, 18] sec$, according to \cref{fig: Consensus Convergence}, the agents have reached to another consensus on the maximum power capacity of the microgrid which is $2100 kW$. Figure \ref{fig: DG1 output power} demonstrates that $P_{1}$ becomes $228.57\,kW$ after the convergence during $[9, 18] sec$. Figure \ref{fig:oPowerDsdisplay} also shows that the output power of the other DGs have increased due to the microgrid capacity reduction. The power sharing ratio $r_i$ for $i=1,2,\cdots,6$ are shown in \cref{fig: Zoomed in Consensus Convergence}. The figure illustrates that, during the transient duration of $(9,15)\,sec$, $r_i$ ratios are not equal. On the contrary, they converge to steady state conditions in $[14,18]\,sec$. Furthermore, from \cref{Fig: Total Output Power}, $P_O$ remains practically equal to $P_L=1600\,kW$. In \cref{fig: DG1 Different Methods}, the results of the proposed control algorithm in \cref{Sec: Dynamic Consensus} is compared with the results of the strategies defined as (\ref{Eq: Strategy 1}) and (\ref{Eq: strategy 3}) in \cref{sec_rev2}. From this figure, it is clear that the output power of $P_1$ obtained from the proposed control algorithm \cref{Sec: Dynamic Consensus} differs from the other two ones during the transient duration. However, after the transient durations of $(3,8)\,sec$ and $(9,15)\,sec$ the output power of $P_1$ from all three methods are the same. Figure \ref{fig: Output Sum Diff Method} demonstrates that the approaches of (\ref{Eq: Strategy 1}) and  (\ref{Eq: strategy 3}) are ineffective to address the load demand. They produce significant deviation in $P_O$ from $P_L$ during transient. On the other hand, upon applying the method of \cref{Sec: Dynamic Consensus}, the deviation drastically reduces, both for increase and decrease in maximum power capacity of DG$_1$.  \vspace{-2mm}

\section{Conclusion}
In this research, the problem of distributed proportional power sharing is studied for microgrids that operate in the grid-connected mode. Firstly, a  consensus algorithm is designed through which, under a variation in the maximum power capacity of a DG, all DGs in the microgrid estimate the updated microgrid capacity. Utilizing the estimations, they generate their output powers in a distributed manner. Stability and convergence of the consensus algorithm are proven. While the consensus algorithm operates in the cyber layer, power commands are sent to the DGs at the physical layer using multiple strategies, discussed in the research. In this regards, practical issues such as ensuring power commands are within acceptable bounds during the transient time of the consensus method, are addressed. However, the consensus algorithm along with the aforementioned strategies does not guarantee maintaining load power during the transient time. Therefore, a modified strategy is proposed to guarantee a match between demanded and delivered power during transient, while the DGs reach a new consensus following a perturbation in grid capacity. The distributed controller is tested in a simulated microgrid. The microgrid is modeled in Matlab/Simulink using the Simscape toolbox. A complete description of the model along with parameters values used for simulation, are given. Simulation results confirm the effectiveness of the proposed strategy. \vspace{-4mm}

\section*{Appendix} \label{sec: appen12}

{\bf Proof of \cref{lem: Lemma2}}

\begin{proof}
From (\ref{Eq: Simplified Consensus 1}), we note that $y_i=s_i-\tilde{P}_T$. Since \cref{lem_rev1} shows that $-(L+\Delta)$ is Hurwitz, therefore from (\ref{Eq: Simplified Consensus 1}) we have, 
\begin{equation} \label{Eq: upper limit of y}
y(t)=e^{-(L+\Delta)t}y(t_0) \; \!\! \Rightarrow \!\! \;
\|y(t)\|\!\!\leq \|e^{-(L+\Delta)t}\| \|y(t_0)\| 
\end{equation}
As explained in \cref{Lem: Lemma 1}, $-(L+\Delta)$ is diagonalizable and all of its eigenvalues are negative and real. 
Assuming $\lambda_1 < 0$ is the largest eigenvalue of $-(L+\Delta)$ and since $y(t_0)=-\delta \mathbf{1}$, we have,
\begin{equation} \label{Eq: upper limit of y vector}
\| y(t)\|\leq e^{\lambda_1t} \| y(t_0) \|=e^{\lambda_1t} \sqrt{N} |\delta | \leq \sqrt{N} |\delta |
\end{equation}	Hence, $\|y\|$ is bounded. 
Since $|y_i| \leq \|y(t)\|$, therefore
\begin{equation} \label{Eq: limit of yit}
|y_i(t)|\leq \sqrt{N}|\delta| \quad \forall \; i=1,2,\cdots,N
\end{equation}
If $|\delta|<{\theta P_T}/{(1+\sqrt{N}})$, then it follows that
\begin{equation} \label{Eq: upper limit of yit}
|y_i(t)| \leq \sqrt{N} |\delta| \leq \sqrt{N} {\theta P_T}/{( 1 + \sqrt{N})}
\end{equation}
and since $y_i=s_i-\tilde{P}_T$, therefore we have
\begin{equation} \label{Eq: limit of si}
\begin{aligned}
\tilde{P}_T-\sqrt{N} {\theta P_T}/&{(1+\sqrt{N}}) \leq s_i(t) \leq \\
&\tilde{P}_T + \sqrt{N} {\theta P_T}/{(1+\sqrt{N}})
\end{aligned}
\end{equation}
Since $\tilde{P}_T = P_T + \delta$, and from the assumption $|\delta|<{\theta P}/{(1+\sqrt{N}})$, we have
\begin{equation} \label{Eq: most expanded inequality of si}
\begin{aligned}
P_T - &{\theta P_T}/{(1 + \sqrt{N}}) - \sqrt{N} {\theta P_T}/{(1 + \sqrt{N}}) \leq s_i(t) \\
&\leq  P_T + {\theta P_T}/{(1 + \sqrt{N}})+ \sqrt{N} {\theta P_T}/{(1 + \sqrt{N}})
\end{aligned}  
\end{equation}
Thus,
\begin{equation} \label{Eq: limit of si_a}
(1-\theta) P_T \leq s_i(t) \leq (1+\theta) P_T 	
\end{equation}

Since for all $t > t_0$, the output power of each DG$_i$ should satisfy $P_i(t)= \frac{P_L}{s_i(t)} P_{i,max} < P_{i,max}$, it is required that $s_i(t) > P_L$ for all $t > t_0$. For guaranteeing $s_i(t) > P_L$, from (\ref{Eq: limit of si_a}), we can impose $(1 - \theta) P_T > P_L$. Therefore, under the dynamics of $\mathbf{S}$ in (\ref{Eq: Matrix Format of the Agents Communication Dynamics}), $P_T > P_L/(1 - \theta)$ or $1 - (P_L/P_T)> \theta $ ensures that $P_i(t) < P_{i,max}$. This completes the proof.
\end{proof}

{\bf Observation on $\delta$:} Lemma \ref{lem: Lemma2} gives the condition $|\delta|<{\theta P_T}/{(1+\sqrt{N}})$ to prevent unfavorable transients in $s_i(t)$. To demonstrate that this condition is not restrictive as $N$ increases, we consider a change in $N$ to $N+1$ and a corresponding change from $P_T$ to $P_T+P_{N+1,max}$. Further, we impose 
\begin{equation} \label{eq: append1}
\frac{\theta(P_T+P_{N+1,max})}{1+\sqrt{N+1}} > \frac{\theta P_T}{1+\sqrt{N}}
\end{equation}
to derive the condition under which $|\delta|$ will increase as we increase $N$ to $N+1$. From (\ref{eq: append1}), we have,
\begin{equation} \label{eq: condition on new DG}
P_{N+1,max}>\Bigl[ \frac{\sqrt{N+1}-\sqrt{N}}{1+\sqrt{N}}\Bigr]P_T  
\end{equation} 
From (\ref{eq: condition on new DG}), it can be observed that $P_{N+1,max}$ can be only a small fraction of $P_T$ to allow $|\delta|$ to increase rather than decrease. For instance, if $N=3$, then $P_{N+1,max}>0.098 P_T$, and if $N=8$, then $P_{N+1,max}>0.045 P_T$ which are small fractions of $P_T$. In addition, comparing the right hand side of (\ref{eq: condition on new DG}) with the average of $P_T$, $P_{T,avg}=P_T/N$, we obtain the minimum ratio of $P_{N+1,max}/P_{T,avg}$ as following
\begin{equation} \label{ratio possible minimum1}
N \Big[ \frac{\sqrt{N+1}-\sqrt{N}}{1+\sqrt{N}}\Big]
\end{equation}
Equation (\ref{ratio possible minimum1}) is strictly less than $\frac{1}{2}$ and it converges to $\frac{1}{2}$ for large values of $N$. This proves that $P_{N+1,max}$ is required to be $P_{N+1,max}\geq (1/2)P_{avg}$ at the worst cases to satisfy the condition on $|\delta|$. Therefore, the condition on $|\delta|$ is not restrictive.



\begin{thebibliography}{00}
\bibitem{rocabert2012control}
J.~Rocabert, A.~Luna, F.~Blaabjerg, and P.~Rodriguez, ``Control of power
converters in ac microgrids,'' {\em IEEE transactions on power electronics},
vol.~27, no.~11, pp.~4734--4749, 2012.

\bibitem{he2017simple}
J.~He, Y.~Pan, B.~Liang, and C.~Wang, ``A simple decentralized islanding
microgrid power sharing method without using droop control,'' {\em IEEE
	Transactions on Smart Grid}, vol.~9, no.~6, pp.~6128--6139, 2017.

\bibitem{du2018distributed}
Y.~Du, X.~Lu, J.~Wang, and S.~Lukic, ``Distributed secondary control strategy
for microgrid operation with dynamic boundaries,'' {\em IEEE Transactions on
	Smart Grid}, 2018.

\bibitem{zhang2012energy}
Y.~Zhang, H.~J. Jia, and L.~Guo, ``Energy management strategy of islanded
microgrid based on power flow control,'' in {\em 2012 IEEE PES Innovative
	Smart Grid Technologies (ISGT)}, pp.~1--8, IEEE, 2012.

\bibitem{han2017review}
Y.~Han, H.~Li, P.~Shen, E.~A.~A. Coelho, and J.~M. Guerrero, ``Review of active
and reactive power sharing strategies in hierarchical controlled
microgrids,'' {\em IEEE Transactions on Power Electronics}, vol.~32, no.~3,
pp.~2427--2451, 2017.

\bibitem{barklund2008energy}
E.~Barklund, N.~Pogaku, M.~Prodanovic, C.~Hernandez-Aramburo, and T.~C. Green,
``Energy management in autonomous microgrid using stability-constrained droop
control of inverters,'' {\em IEEE Transactions on Power Electronics},
vol.~23, no.~5, pp.~2346--2352, 2008.

\bibitem{deng2016enhanced}
Y.~Deng, Y.~Tao, G.~Chen, G.~Li, and X.~He, ``Enhanced power flow control for
grid-connected droop-controlled inverters with improved stability,'' {\em
	IEEE Transactions on Industrial Electronics}, vol.~64, no.~7, pp.~5919--5929,
2016.

\bibitem{gui2018improved}
Y.~Gui, C.~Kim, C.~C. Chung, J.~M. Guerrero, Y.~Guan, and J.~C. Vasquez,
``Improved direct power control for grid-connected voltage source
converters,'' {\em IEEE Transactions on Industrial Electronics}, vol.~65,
no.~10, pp.~8041--8051, 2018.

\bibitem{fan2016distributed}
Y.~Fan, G.~Hu, and M.~Egerstedt, ``Distributed reactive power sharing control
for microgrids with event-triggered communication,'' {\em IEEE Transactions
	on Control Systems Technology}, vol.~25, no.~1, pp.~118--128, 2016.

\bibitem{he2014consensus}
D.~He, D.~Shi, and R.~Sharma, ``Consensus-based distributed cooperative control
for microgrid voltage regulation and reactive power sharing,'' in {\em IEEE
	PES Innovative Smart Grid Technologies, Europe}, pp.~1--6, IEEE, 2014.


\bibitem{mahmud2014robust}
M.~Mahmud, M.~Hossain, H.~Pota, and A.~Oo, ``Robust nonlinear distributed
controller design for active and reactive power sharing in islanded
microgrids,'' {\em IEEE Transactions on Energy Conversion}, vol.~29, no.~4,
pp.~893--903, 2014.

\bibitem{guerrero2005output}
J.~M. Guerrero, L.~G. De~Vicuna, J.~Matas, M.~Castilla, and J.~Miret, ``Output
impedance design of parallel-connected ups inverters with wireless
load-sharing control,'' {\em IEEE Transactions on industrial electronics},
vol.~52, no.~4, pp.~1126--1135, 2005.

\bibitem{guerrero2007decentralized}
J.~M. Guerrero, J.~Matas, L.~G. de~Vicuna, M.~Castilla, and J.~Miret,
``Decentralized control for parallel operation of distributed generation
inverters using resistive output impedance,'' {\em IEEE Transactions on
	industrial electronics}, vol.~54, no.~2, pp.~994--1004, 2007.

\bibitem{aalipour2018proportional}
F.~Aalipour and T.~Das, ``Proportional power sharing consensus in distributed
generators,'' in {\em ASME 2018 Dynamic Systems and Control Conference},
pp.~V002T17A002--V002T17A002, American Society of Mechanical Engineers, 2018.

\bibitem{chen2015distributed}
G.~Chen, F.~L. Lewis, E.~N. Feng, and Y.~Song, ``Distributed optimal active
power control of multiple generation systems,'' {\em IEEE Transactions on
	Industrial Electronics}, vol.~62, no.~11, pp.~7079--7090, 2015.

\bibitem{zhong2011robust}
Q.-C. Zhong, ``Robust droop controller for accurate proportional load sharing
among inverters operated in parallel,'' {\em IEEE Transactions on Industrial
	Electronics}, vol.~60, no.~4, pp.~1281--1290, 2011.

\bibitem{pantoja2011population}
A.~Pantoja and N.~Quijano, ``A population dynamics approach for the dispatch of
distributed generators,'' {\em IEEE Transactions on Industrial Electronics},
vol.~58, no.~10, pp.~4559--4567, 2011.

\bibitem{lin1984hierarchical}
C.-E. Lin and G.~Viviani, ``Hierarchical economic dispatch for piecewise
quadratic cost functions,'' {\em IEEE transactions on power apparatus and
	systems}, no.~6, pp.~1170--1175, 1984.

\bibitem{kar2012distributed}
S.~Kar and G.~Hug, ``Distributed robust economic dispatch in power systems: A
consensus+ innovations approach,'' in {\em 2012 IEEE Power and Energy Society
	General Meeting}, pp.~1--8, IEEE, 2012.

\bibitem{lou2016distributed}
G.~Lou, W.~Gu, Y.~Xu, M.~Cheng, and W.~Liu, ``Distributed mpc-based secondary
voltage control scheme for autonomous droop-controlled microgrids,'' {\em
	IEEE transactions on sustainable energy}, vol.~8, no.~2, pp.~792--804, 2016.

\bibitem{vaccaro2011decentralized}
A.~Vaccaro, G.~Velotto, and A.~F. Zobaa, ``A decentralized and cooperative
architecture for optimal voltage regulation in smart grids,'' {\em IEEE
	Transactions on Industrial Electronics}, vol.~58, no.~10, pp.~4593--4602,
2011.

\bibitem{liu2018game}
W.~Liu, W.~Gu, J.~Wang, W.~Yu, and X.~Xi, ``Game theoretic non-cooperative
distributed coordination control for multi-microgrids,'' {\em IEEE
	Transactions on Smart Grid}, vol.~9, no.~6, pp.~6986--6997, 2018.

\bibitem{chaudhuri2004wide}
B.~Chaudhuri, R.~Majumder, and B.~C. Pal, ``Wide-area measurement-based
stabilizing control of power system considering signal transmission delay,''
{\em IEEE Transactions on Power Systems}, vol.~19, no.~4, pp.~1971--1979,
2004.

\bibitem{wu2004evaluation}
H.~Wu, K.~S. Tsakalis, and G.~T. Heydt, ``Evaluation of time delay effects to
wide-area power system stabilizer design,'' {\em IEEE Transactions on Power
	Systems}, vol.~19, no.~4, pp.~1935--1941, 2004.

\bibitem{xiao2004fast}
L.~Xiao and S.~Boyd, ``Fast linear iterations for distributed averaging,'' {\em
	Systems \& Control Letters}, vol.~53, no.~1, pp.~65--78, 2004.

\bibitem{dorfler2018electrical}
F.~D{\"o}rfler, J.~W. Simpson-Porco, and F.~Bullo, ``Electrical networks and
algebraic graph theory: Models, properties, and applications,'' {\em
 Proceedings of the IEEE}, vol.~106, no.~5, pp.~977--1005, 2018.


\bibitem{charalambous2015distributed}
T.~Charalambous, Y.~Yuan, T.~Yang, W.~Pan, C.~N. Hadjicostis, and M.~Johansson,
``Distributed finite-time average consensus in digraphs in the presence of
time delays,'' {\em IEEE Transactions on Control of Network Systems}, vol.~2,
no.~4, pp.~370--381, 2015.

\bibitem{khalil2002nonlinear}
H.~K. Khalil, {\em Nonlinear systems}, vol.~3.
\newblock Prentice hall Upper Saddle River, NJ, 2002.

\bibitem{mesbahi2010graph}
M.~Mesbahi and M.~Egerstedt, {\em Graph theoretic methods in multiagent
	networks}.
\newblock Princeton University Press, 2010.


\bibitem{horn1990matrix}
R.~A. Horn and C.~R. Johnson, {\em Matrix analysis}, vol.~2.
\newblock Cambridge University Press, 2013.

\bibitem{aalipour2017distributed}
F.~Aalipour, A.~Gusrialdi, and Z.~Qu, ``Distributed optimal output feedback
control of heterogeneous multi-agent systems under a directed graph,'' {\em
	IFAC-PapersOnLine}, vol.~50, no.~1, pp.~5097--5102, 2017.

\bibitem{schiffer2015voltage}
J.~Schiffer, T.~Seel, J.~Raisch, and T.~Sezi, ``Voltage stability and reactive
  power sharing in inverter-based microgrids with consensus-based distributed
  voltage control,'' {\em IEEE Transactions on Control Systems Technology},
  vol.~24, no.~1, pp.~96--109, 2015.
  
\bibitem{cai2018distributed}
H.~Cai and G.~Hu, ``Distributed robust hierarchical power sharing control of
grid-connected spatially concentrated ac microgrid,'' {\em IEEE Transactions
	on Control Systems Technology}, vol.~27, no.~3, pp.~1012--1022, 2018.


\bibitem{Nasirian2016Droop}
V.~{Nasirian}, Q.~{Shafiee}, J.~M. {Guerrero}, F.~L. {Lewis}, and A.~{Davoudi},
  ``Droop-free distributed control for ac microgrids,'' {\em IEEE Transactions
  on Power Electronics}, vol.~31, no.~2, pp.~1600--1617, 2016.

\bibitem{Matlab_Sims}
``250-kw grid-connected pv array.'' https://www.mathworks.com/help/physmod/sps/examples/250-kw-grid-connected-pv-array.html.




\end{thebibliography}
\end{document}